\documentclass[10pt]{article}

\title{Accurate computation of Galerkin double surface integrals in the 3-D boundary element method}
\author{Ross Adelman\footnote{rna@umd.edu, also at Department of Computer Science, University of Maryland, College Park}, Nail A. Gumerov\footnote{gumerov@umiacs.umd.edu, also at Fantalgo, LLC}, and Ramani Duraiswami\footnote{ramani@umiacs.umd.edu, also at Department of Computer Science, University of Maryland, College Park and Fantalgo, LLC}\\Institute for Advanced Computer Studies, University of Maryland, College Park}
\date{}

\usepackage{amsmath}
\usepackage{amsfonts}
\usepackage{amsthm}
\usepackage[bf,labelsep=period]{caption}
\usepackage{fancyhdr}
\usepackage[margin=1.0in]{geometry}
\usepackage{microtype}
\usepackage[pdftex]{graphicx}
\usepackage{times}
\usepackage{verbatim}

\theoremstyle{definition}
\newtheorem{theorem}{Theorem}
\theoremstyle{plain}

\makeatletter
\renewenvironment{proof}[1][\proofname] {\par\pushQED{\qed}\normalfont\topsep6\p@\@plus6\p@\relax\trivlist\item[\hskip\labelsep\bfseries#1\@addpunct{.}]\ignorespaces}{\popQED\endtrivlist\@endpefalse}
\makeatother

\begin{document}

\maketitle

\section*{Abstract}

Many boundary element integral equation kernels are based on the Green's functions of the Laplace and Helmholtz equations in three dimensions.
These include, for example, the Laplace, Helmholtz, elasticity, Stokes, and Maxwell's equations.
Integral equation formulations lead to more compact, but dense linear systems.
These dense systems are often solved iteratively via Krylov subspace methods, which may be accelerated via the fast multipole method.
There are advantages to Galerkin formulations for such integral equations, as they treat problems associated with kernel singularity, and lead to symmetric and better conditioned matrices.
However, the Galerkin method requires each entry in the system matrix to be created via the computation of a double surface integral over one or more pairs of triangles.
There are a number of semi-analytical methods to treat these integrals, which all have some issues, and are discussed in this paper.
We present novel methods to compute all the integrals that arise in Galerkin formulations involving kernels based on the Laplace and Helmholtz Green's functions to any specified accuracy.
Integrals involving completely geometrically separated triangles are non-singular and are computed using a technique based on spherical harmonics and multipole expansions and translations, which results in the integration of polynomial functions over the triangles.
Integrals involving cases where the triangles have common vertices, edges, or are coincident are treated via scaling and symmetry arguments, combined with automatic recursive geometric decomposition of the integrals.
Example results are presented, and the developed software is available as open source.

\section{Introduction}

The Galerkin boundary element method (BEM) is a powerful method for solving integral equations involving kernels based on the Laplace and/or Helmholtz equations' Green's functions in three dimensions \cite{harrington1968, pozrikids2002, gibson2008}.
When the boundary is discretized using triangular elements, constructing the system matrix requires computing double surface integrals over pairs of these triangles.
Because the kernels being integrated are singular, these integrals can be difficult to compute, especially when the two triangles are proximate, share a vertex, an edge, or are the same.
Depending on the relative geometry of the two triangles, there are many different methods for computing them.
For example, when the two triangles do not touch, the integral is completely regular and can be computed accurately via numerical means, e.g., Gaussian quadrature \cite{arcioni1997}.
Semi-analytical methods, where the inside integral is computed analytically and the outside integral is computed numerically, have also been proposed \cite{aimi2002, wang2003}.

However, when the two triangles share a vertex, an edge, or are the same, things become more complicated.
There are analytical expressions for the case when the two triangles are the same \cite{eibert1995}, but not for when they share only a vertex or an edge.
In these cases, the semi-analytical methods do not always work.
This is because, depending on the kernel being integrated, the inside integral can be hypersingular.
While there are analytical expressions available for them, they are singular along the corners and edges of the corresponding triangle.
When the two triangles share a vertex or an edge, these singularities are included in the outside integral.
The usual semi-analytical methods will not work in these cases because they are not designed to properly handle the singularities.

The double integrals are weakly singular, so while the inside integrals may be hypersingular and the expressions for them may be singular in some places, they are completely integrable.
Nevertheless, actually integrating them in practice can be hard.
Therefore, more sophisticated semi-analytical methods have been developed over the years.
These include:
singularity subtraction and ``to the boundary'' techniques \cite{gray2001, salvadori2001, jarvenpaa2006, fata2009, vipiana2013};
specialized quadrature methods that are designed for the singularities involved, such as those based on the double exponential formula \cite{polimeridis2010, lopez2011}; and
other regularization methods, such as the Duffy transformation \cite{mousavi2010, fata2010, polimeridis2011}.
Many of these methods work very well, but because they all attempt to tackle the singularity issue directly, their analysis is very involved.

In this paper, we present a method for computing the integrals that completely avoids the computation of singular integrals.
The approach relies on several scaling properties of the integrals and the kernels being integrated.
When two triangles share a vertex, an edge, or are the same, the integral is decomposed into several smaller integrals, some of which are related back to the original integral using simple analysis.
This is done in such a way so that only regular integrals need to be computed explicitly.
Any integrals involving singularities are computed implicitly during the procedure.
The regular integrals can be computed using standard semi-analytical methods, but in this paper, we also present an analytical method for doing so.
This method uses spherical harmonics and multipole and local expansions and translations.
The only source of error in the method is from truncating these expansions.
However, this error is precisely controlled by choosing the appropriate truncation number or recursively subdiving the problem.
Although we developed these methods primarily for kernels related to the Laplace equation's Green's function, we also describe how they can be extended to kernels related to the Helmholtz equation's Green's function.
Finally, we have implemented a Galerkin BEM library in MATLAB for the Laplace equation using these two methods.
We show some example problems and provide some error analysis.
We have made this library freely available for download \cite{galerkin_webpage}.
The library can be used to recreate all the examples seen in this paper or to create entirely new ones.

\section{Background}
\label{background_sec}

\begin{figure}[t]
	\centering
	\includegraphics[scale=0.75]{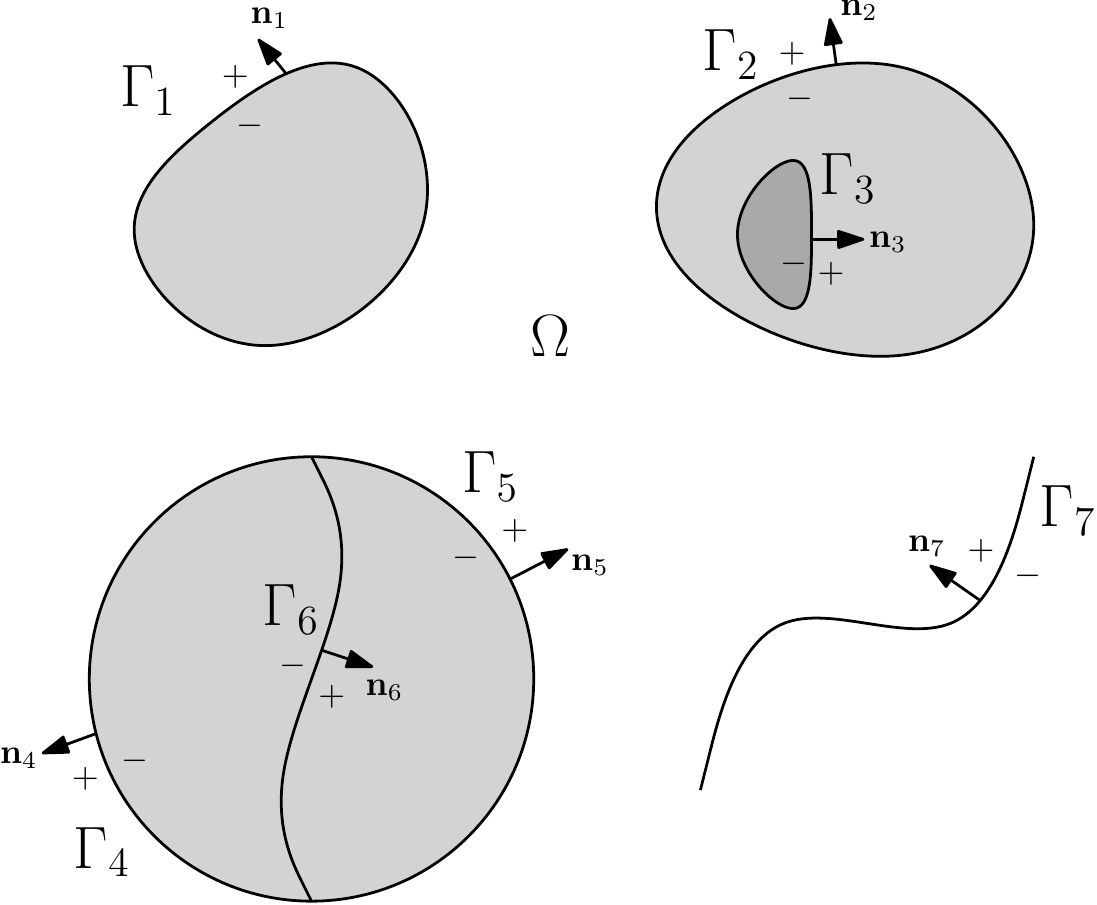}
	\caption{Several different types of boundaries: $\Gamma_1$ is closed; $\Gamma_2$ is closed, but contains another boundary, $\Gamma_3$, also closed; $\Gamma_4$ and $\Gamma_5$ form a closed region, but this region is divided into two by $\Gamma_6$; and $\Gamma_7$ is open.}
	\label{example_geometry}
\end{figure}

The Galerkin BEM is a powerful method for solving integral equations related to the Laplace or Helmholtz Green's functions in three dimensions.
Consider the following boundary value problem (BVP):
\begin{equation}
\nabla^2\phi\left(\mathbf{x}\right) = 0\quad\text{or}\quad\nabla^2\phi\left(\mathbf{x}\right) + k^2\phi\left(\mathbf{x}\right) = 0,\quad\mathbf{x} \in \Omega,
\end{equation}
where $k > 0$ is the wavenumber in the Helmholtz equation.
General boundary conditions are enforced on both sides of the boundary, and are given by:
\begin{equation}
\alpha_1^+\left(\mathbf{x}\right)\phi^+\left(\mathbf{x}\right) +
\beta_1^+\left(\mathbf{x}\right)q^+\left(\mathbf{x}\right) +
\alpha_1^-\left(\mathbf{x}\right)\phi^-\left(\mathbf{x}\right) +
\beta_1^-\left(\mathbf{x}\right)q^-\left(\mathbf{x}\right) =
\gamma_1\left(\mathbf{x}\right),\quad\mathbf{x} \in \Gamma,
\end{equation}
\begin{equation}
\alpha_2^+\left(\mathbf{x}\right)\phi^+\left(\mathbf{x}\right) +
\beta_2^+\left(\mathbf{x}\right)q^+\left(\mathbf{x}\right) +
\alpha_2^-\left(\mathbf{x}\right)\phi^-\left(\mathbf{x}\right) +
\beta_2^-\left(\mathbf{x}\right)q^-\left(\mathbf{x}\right) =
\gamma_2\left(\mathbf{x}\right),\quad\mathbf{x} \in \Gamma.
\end{equation}
In addition, boundary conditions must be set at infinity.
For the Laplace equation, the potential should decay to zero at large distances.
For the Helmholtz equation, the potential should decay to zero at large distances as well, but should also be composed of outgoing waves only.
The Sommerfeld radiation condition provides such a constraint:
\begin{equation}
\lim_{\left|\mathbf{x}\right| \rightarrow \infty}\phi\left(\mathbf{x}\right) = 0,\quad
\lim_{\left|\mathbf{x}\right| \rightarrow \infty}\left|\mathbf{x}\right|\left(\frac{d\phi}{d\left|\mathbf{x}\right|}\left(\mathbf{x}\right) - ik\phi\left(\mathbf{x}\right)\right) = 0.
\end{equation}

The boundaries can be closed or open.
See Fig.\ \ref{example_geometry} for some example boundaries.
For closed boundaries, the interior side of the boundary is the ``$-$'' side and the exterior side is the ``$+$'' side.
For open boundaries, since there is no inside or outside, designating each side of the boundary as ``$+$'' or ``$-$'' can be done arbitrarily.
The values, $\phi^+$ and $\phi^-$, are the potential on the ``$+$'' and ``$-$'' sides of the surface, respectively.
Likewise, the values, $q^+$ and $q^-$, are the normal derivatives of the potential on the ``$+$'' and ``$-$'' sides of the surface, respectively, and are given by
\begin{equation}
q^\pm = \frac{\partial\phi^\pm}{\partial\mathbf{n}^\pm} = \left(\mathbf{n}^\pm\cdot\nabla_{\mathbf{x}}\right)\phi^\pm,
\end{equation}
where $\mathbf{n}^\pm = \mp\mathbf{n}$ (i.e., $\mathbf{n}^+$ goes from the ``$+$'' side to the ``$-$'' side, and vice versa).

To solve the BVP, we use an indirect boundary integral formulation called the layer potential formuation.
Using Green's theorem, the Laplace or Helmholtz equation is transformed from a differential equation into an integral equation:
\begin{equation}
\label{intgrl_form}
\phi\left(\mathbf{x}\right) = L\left[\sigma\right]\left(\mathbf{x}\right) + M\left[\mu\right]\left(\mathbf{x}\right),
\end{equation}
where
\begin{equation}
L\left[\sigma\right]\left(\mathbf{x}\right) = \int_{\mathbf{x}^\prime \in \Gamma}\sigma\left(\mathbf{x}^\prime\right)G\left(\mathbf{x} - \mathbf{x}^\prime\right)dS\left(\mathbf{x}^\prime\right),
\end{equation}
\begin{equation}
M\left[\mu\right]\left(\mathbf{x}\right) = \int_{\mathbf{x}^\prime \in \Gamma}\mu\left(\mathbf{x}^\prime\right)\left(\mathbf{n}^\prime\cdot\nabla_{\mathbf{x}^\prime}\right)G\left(\mathbf{x} - \mathbf{x}^\prime\right)dS\left(\mathbf{x}^\prime\right)
\end{equation}
are the single- and double-layer potentials \cite{burton1971}, and
\begin{equation}
G\left(\mathbf{r}\right) = \frac{1}{4\pi\left|\mathbf{r}\right|}\quad\text{or}\quad{}G\left(\mathbf{r}\right) = \frac{\exp\left(ik\left|\mathbf{r}\right|\right)}{4\pi\left|\mathbf{r}\right|}
\end{equation}
is the Laplace or Helmholtz equation's Green's function.
The single-layer potential, $L\left[\sigma\right]\left(\mathbf{x}\right)$, is the potential due to the monopole source density distribution, $\sigma\left(\mathbf{x}^\prime\right)$, on the boundary.
Likewise, the double-layer potential, $M\left[\mu\right]\left(\mathbf{x}\right)$, is the potential due to the dipole source density distribution, $\mu\left(\mathbf{x}^\prime\right)$, on the boundary.
In the differential equation, we seek a solution to the potential governed by the Laplace equation.
However, in the integral equation, we seek the source density distributions, $\sigma\left(\mathbf{x}^\prime\right)$ and $\mu\left(\mathbf{x}^\prime\right)$, on the boundary that give rise to that potential.
The advantage of the BEM is that the expression in Eq.\ (\ref{intgrl_form}) relating the source density distributions back to the potential automatically satisfies the original differential equation.
Moreover, the Green's functions satisfy the boundary conditions at infinity, so as long as the source density distributions are bounded and finite, the single- and double-layer potentials will satsify them as well.
Thus, we need only concern ourselves with searching for the source density distributions that satisfy the remaining boundary conditions.
To do this, we need to express the potentials and normal derivatives on either side of the boundary in terms of the source density distributions.
Jump conditions provide such a relationship:
\begin{equation}
\label{phipm}
\phi^\pm\left(\mathbf{x}\right) = L\left[\sigma\right]\left(\mathbf{x}\right) + M\left[\mu\right]\left(\mathbf{x}\right) \pm \frac{1}{2}\mu\left(\mathbf{x}\right),
\end{equation}
\begin{equation}
\label{qpm}
q^\pm\left(\mathbf{x}\right) = \mp{}L^\prime\left[\sigma\right]\left(\mathbf{x}\right) \mp M^\prime\left[\mu\right]\left(\mathbf{x}\right) + \frac{1}{2}\sigma\left(\mathbf{x}\right),
\end{equation}
where
\begin{equation}
L^\prime\left[\sigma\right]\left(\mathbf{x}\right) = \left(\mathbf{n}\cdot\nabla_{\mathbf{x}}\right)\int_{\mathbf{x}^\prime \in \Gamma}\sigma\left(\mathbf{x}^\prime\right)G\left(\mathbf{x} - \mathbf{x}^\prime\right)dS\left(\mathbf{x}^\prime\right),
\end{equation}
\begin{equation}
M^\prime\left[\mu\right]\left(\mathbf{x}\right) = \left(\mathbf{n}\cdot\nabla_{\mathbf{x}}\right)\int_{\mathbf{x}^\prime \in \Gamma}\mu\left(\mathbf{x}^\prime\right)\left(\mathbf{n}^\prime\cdot\nabla_{\mathbf{x}^\prime}\right)G\left(\mathbf{x} - \mathbf{x}^\prime\right)dS\left(\mathbf{x}^\prime\right).
\end{equation}
Plugging Eqs.\ (\ref{phipm}) and (\ref{qpm}) into the boundary conditions and rearranging,
\begin{equation}
a_1\left(L\left[\sigma\right] + M\left[\mu\right]\right) +
b_1\left(L^\prime\left[\sigma\right] + M^\prime\left[\mu\right]\right) +
c_1\sigma +
d_1\mu =
\gamma_1,
\end{equation}
\begin{equation}
a_2\left(L\left[\sigma\right] + M\left[\mu\right]\right) +
b_2\left(L^\prime\left[\sigma\right] + M^\prime\left[\mu\right]\right) +
c_2\sigma +
d_2\mu =
\gamma_2,
\end{equation}
where
\begin{equation}
a_1 = \alpha_1^+ + \alpha_1^-
,\quad
b_1 = -\beta_1^+ + \beta_1^-
,\quad
c_1 = \frac{1}{2}\left(\beta_1^+ + \beta_1^-\right)
,\quad
d_1 = \frac{1}{2}\left(\alpha_1^+ - \alpha_1^-\right),
\end{equation}
\begin{equation}
a_2 = \alpha_2^+ + \alpha_2^-
,\quad
b_2 = -\beta_2^+ + \beta_2^-
,\quad
c_2 = \frac{1}{2}\left(\beta_2^+ + \beta_2^-\right)
,\quad
d_2 = \frac{1}{2}\left(\alpha_2^+ - \alpha_2^-\right).
\end{equation}
In these expressions, we have dropped the argument, $\mathbf{x}$, to save space (i.e., $\sigma\left(\mathbf{x}\right)$ becomes $\sigma$).
In order to make the problem computationally tractable, the source density distributions, $\sigma\left(\mathbf{x}^\prime\right)$ and $\mu\left(\mathbf{x}^\prime\right)$, are each written as a linear combination of $N$ basis functions:
\begin{equation}
\sigma\left(\mathbf{x}^\prime\right) = \sum_{j = 1}^N\sigma_jf_j\left(\mathbf{x}^\prime\right),\quad
\mu\left(\mathbf{x}^\prime\right) = \sum_{j = 1}^N\mu_jf_j\left(\mathbf{x}^\prime\right).
\end{equation}
For constant triangular elements, there is one basis function per triangle that is equal to one on that element and zero everywhere else.
For linear triangular elements, there is one basis function per vertex that is equal to one at that vertex, zero at all the other vertices, and piecewise linear everywhere else.
We need to compute the coefficients of these basis functions so that the boundary conditions are satisfied.
In other words, we seek $\sigma_1, \sigma_2, \ldots, \alpha_N$, $\mu_1, \mu_2, \ldots,$ and $\mu_N$ such that
\begin{equation}
\sum_{j = 1}^N\sigma_j\left(a_1L\left[f_j\right] +
b_1L^\prime\left[f_j\right] +
c_1f_j\right) +
\sum_{j = 1}^N\mu_j\left(a_1M\left[f_j\right] +
b_1M^\prime\left[f_j\right] +
d_1\right) =
\gamma_1,
\end{equation}
\begin{equation}
\sum_{j = 1}^N\sigma_j\left(a_2L\left[f_j\right] +
b_2L^\prime\left[f_j\right] +
c_2f_j\right) +
\sum_{j = 1}^N\mu_j\left(a_2M\left[f_j\right] +
b_2M^\prime\left[f_j\right] +
d_2\right) =
\gamma_2.
\end{equation}
Or, more compactly,
\begin{equation}
\sum_{j = 1}^N\sigma_jA_1\left[f_j\right] +
\sum_{j = 1}^N\mu_jB_1\left[f_j\right] =
\gamma_1,
\end{equation}
\begin{equation}
\sum_{j = 1}^N\sigma_jA_2\left[f_j\right] +
\sum_{j = 1}^N\mu_jB_2\left[f_j\right] =
\gamma_2,
\end{equation}
where
\begin{equation}
A_1\left[f_j\right] =
a_1L\left[f_j\right] +
b_1L^\prime\left[f_j\right] +
c_1f_j
,\quad
B_1\left[f_j\right] =
a_1M\left[f_j\right] +
b_1M^\prime\left[f_j\right] +
d_1,
\end{equation}
\begin{equation}
A_2\left[f_j\right] =
a_2L\left[f_j\right] +
b_2L^\prime\left[f_j\right] +
c_2f_j
,\quad
B_2\left[f_j\right] =
a_2M\left[f_j\right] +
b_2M^\prime\left[f_j\right] +
d_2.
\end{equation}

The two most commonly used methods for enforcing the boundary conditions are the collocation method and the Galerkin method.
The collocation method works by enforcing the boundary conditions at $N$ matching points:
\begin{equation}
\sum_{j = 1}^N\sigma_jA_1\left[f_j\right]\left(\mathbf{x}_i\right) +
\sum_{j = 1}^N\mu_jB_1\left[f_j\right]\left(\mathbf{x}_i\right) =
\gamma_1\left(\mathbf{x}_i\right),\quad{}i = 1, 2, \ldots, N,
\end{equation}
\begin{equation}
\sum_{j = 1}^N\sigma_jA_2\left[f_j\right]\left(\mathbf{x}_i\right) +
\sum_{j = 1}^N\mu_jB_2\left[f_j\right]\left(\mathbf{x}_i\right) =
\gamma_2\left(\mathbf{x}_i\right),\quad{}i = 1, 2, \ldots, N.
\end{equation}
The method is so named because these points are typically collocated with the modeling elements (e.g., in the case of constant triangular elements, one is placed at every triangle's centroid).
Collocation methods have long been used \cite{rao1979, rao1982}.
They are easy to understand, and the integral expressions necessary for implementing them have been studied and derived by many authors.
This includes, for example, piecewise constant and linear basis functions on triangular elements, which are the most commonly used \cite{wilton1984, davey1989, graglia1993, katz2001}.
However, they suffer from a few problems.
Many of the boundary integrals are hypersingular, which make them hard (or sometimes even impossible) to compute, especially for points on the corners or edges of the boundary.

The Galerkin method overcomes these problems by enforcing the boundary conditions in an integral sense.
The boundary integral equation is multiplied by each of the same $N$ basis functions from before and integrated over the boundary a second time:
\begin{equation}
\sum_{j = 1}^N\sigma_j\int_{\mathbf{x} \in \Gamma}f_iA_1\left[f_j\right]dS\left(\mathbf{x}\right) +
\sum_{j = 1}^N\mu_j\int_{\mathbf{x} \in \Gamma}f_iB_1\left[f_j\right]dS\left(\mathbf{x}\right) =
\int_{\mathbf{x} \in \Gamma}f_i\gamma_1dS\left(\mathbf{x}\right),\quad{}i = 1, 2, \ldots, N,
\end{equation}
\begin{equation}
\sum_{j = 1}^N\sigma_j\int_{\mathbf{x} \in \Gamma}f_iA_2\left[f_j\right]dS\left(\mathbf{x}\right) +
\sum_{j = 1}^N\mu_j\int_{\mathbf{x} \in \Gamma}f_iB_2\left[f_j\right]dS\left(\mathbf{x}\right) =
\int_{\mathbf{x} \in \Gamma}f_i\gamma_2dS\left(\mathbf{x}\right),\quad{}i = 1, 2, \ldots, N.
\end{equation}
By doing so, all the hypersingular integrals become weakly singular.
Morever, the system matrices in the Galerkin method are typically symmetric, better conditioned, and have better convergence properties \cite{natarajan1998, ademoyero2001}.
However, the extra integral over the boundary complicates the computation of the entries in the system matrix.

\section{Double Surface Integrals}
\label{double_surf_sec}

\begin{figure}[t]
	\centering
	\includegraphics[scale=0.75]{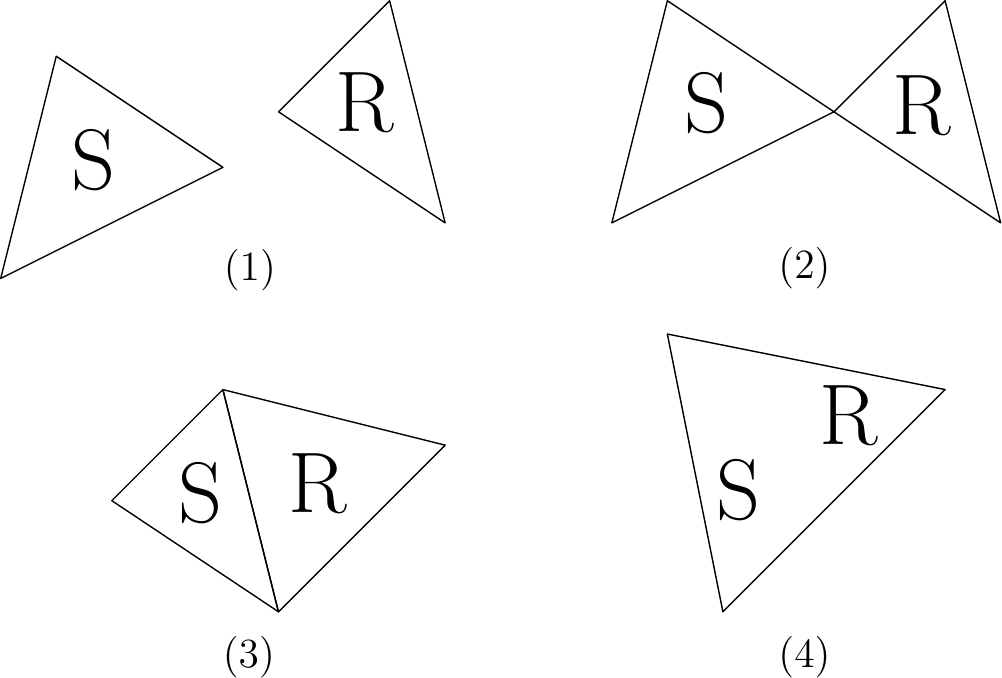}
	\caption{In practice, the relative geometry of two triangles is one of the following: (1) the two triangles do not touch; (2) they share a vertex; (3) they share an edge; or (4) they are the same.}
	\label{four_geometries}
\end{figure}

When the boundary is discretized using linear triangular elements, the double surface integrals are performed over pairs of these triangles.
In each pair, one is called the ``source'' triangle, and the other the ``receiver'' triangle.
The inside integral is over the source triangle, $\text{S}$, and the outside integral is over the receiver triangle, $\text{R}$.
Thus, when populating the system matrix, we need to compute integrals of the following form:
\begin{equation}
I = \int_{\mathbf{x} \in \text{R}}
\left(\sigma_0 + \mathbf{p}\cdot\mathbf{x}\right)
\int_{\mathbf{x}^\prime \in \text{S}}\left(\sigma_0^\prime + \mathbf{p}^\prime\cdot\mathbf{x}^\prime\right)F\left(\mathbf{x} - \mathbf{x}^\prime\right)dS\left(\mathbf{x}^\prime\right)dS\left(\mathbf{x}\right),
\end{equation}
where $\sigma_0^\prime + \mathbf{p}^\prime\cdot\mathbf{x}^\prime$ is the source density distribution over the source triangle, $\sigma_0 + \mathbf{p}\cdot\mathbf{x}$ is the weight function over the receiver triangle, and $F\left(\mathbf{r}\right)$ is the kernel being integrated.
To implement the Galerkin BEM described in Sec.\ \ref{background_sec}, we need to compute this integral for the following four kernels:
\begin{equation}
\label{F1eq}
F_1\left(\mathbf{x} - \mathbf{x}^\prime\right) = G\left(\mathbf{x} - \mathbf{x}^\prime\right),
\end{equation}
\begin{equation}
\label{F2eq}
F_2\left(\mathbf{x} - \mathbf{x}^\prime\right) = \left(\mathbf{n}^\prime\cdot\nabla_{\mathbf{x}^\prime}\right)G\left(\mathbf{x} - \mathbf{x}^\prime\right),
\end{equation}
\begin{equation}
\label{F3eq}
F_3\left(\mathbf{x} - \mathbf{x}^\prime\right) = \left(\mathbf{n}\cdot\nabla_{\mathbf{x}}\right)G\left(\mathbf{x} - \mathbf{x}^\prime\right),
\end{equation}
\begin{equation}
\label{F4eq}
F_4\left(\mathbf{x} - \mathbf{x}^\prime\right) = \left(\mathbf{n}\cdot\nabla_{\mathbf{x}}\right)\left(\mathbf{n}^\prime\cdot\nabla_{\mathbf{x}^\prime}\right)G\left(\mathbf{x} - \mathbf{x}^\prime\right),
\end{equation}
where $F_1\left(\mathbf{r}\right)$ and $F_2\left(\mathbf{r}\right)$ correspond to the single- and double-layer potentials, and $F_3\left(\mathbf{r}\right)$ and $F_4\left(\mathbf{r}\right)$ correspond to their normal derivatives.

Computing the integral for these four kernels for all commonly encountered geometries is the focus of this paper.
In practice, the relative geometry of the two triangles is one of the following:
(1) the two triangles do not touch;
(2) they share a vertex;
(3) they share an edge; or
(4) they are the same (see Fig.\ \ref{four_geometries}).
These are called the zero-, one-, two-, and three-touch cases, respectively.
In this naming scheme, the number represents how many vertices the two triangles share.
For example, in the one-touch case, the two triangles share a single vertex.

\section{Reduction of the Helmholtz Kernel to the Laplace Kernel}

\begin{figure}[t]
	\centering
	\includegraphics[scale=0.75]{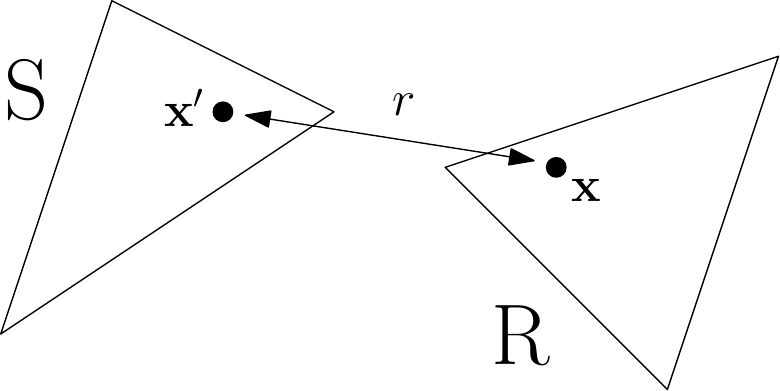}
	\caption{For the double surface integral of the Helmholtz equation's Green's function, assume the two triangles are small and close to each other.  In other words, for any two points, one on the source triangle and one on the receiver triangle, assume the distance between them, $r$, is less than $1 / k$.}
	\label{krlessthanone}
\end{figure}

In order to apply the Galerkin BEM to integral equations whose kernels are based on the Helmholtz Green's function, we need to compute double surface integrals of the four kernels given in Eqs.\ (\ref{F1eq}) - (\ref{F4eq}) over pairs of triangles.
The remainder of this paper describes several novel methods for computing these integrals when the Green's function is that of the Laplace equation.
These methods are not directly applicable to the Helmholtz Green's function.
However, they can be extended to the Helmholtz equation (and any other problem involving the Laplace and/or Helmholtz kernels, such as elastostatics, Stokes problems, and Maxwell's equations) via the singularity subtraction technique.

To reduce the Helmholtz kernel to the Laplace kernel, we use the singularity subtraction technique \cite{jarvenpaa2006}.
As an example, consider the following double surface integral over a pair of triangles that would appear in a Galerkin BEM for the Helmholtz equation:
\begin{equation}
I = \int_{\mathbf{x} \in \text{R}}\int_{\mathbf{x}^\prime \in \text{S}}\frac{\exp\left(ikr\right)}{4\pi{}r}dS\left(\mathbf{x}^\prime\right)dS\left(\mathbf{x}\right),
\end{equation}
where $r = \left|\mathbf{x} - \mathbf{x}^\prime\right|$.
Assume the two triangles are small and close to each other (i.e., $kr < 1$, see Fig.\ \ref{krlessthanone}).
The singularity subtraction technique works by separating the Helmholtz kernel into two pieces, a singular part (which happens to be the Laplace kernel) and a regular part.
This is done by expanding the numerator in the Helmholtz equation as a Taylor series expansion around $r = 0$:
\begin{equation}
I = \int_{\mathbf{x} \in \text{R}}\int_{\mathbf{x}^\prime \in \text{S}}\left(\sum_{n = 0}^\infty\frac{\left(ikr\right)^n}{n!}\right)\frac{1}{4\pi{}r}dS\left(\mathbf{x}^\prime\right)dS\left(\mathbf{x}\right),
\end{equation}
\begin{equation}
I = \int_{\mathbf{x} \in \text{R}}\int_{\mathbf{x}^\prime \in \text{S}}\frac{1}{4\pi{}r}dS\left(\mathbf{x}^\prime\right)dS\left(\mathbf{x}\right) +
\int_{\mathbf{x} \in \text{R}}\int_{\mathbf{x}^\prime \in \text{S}}\frac{1}{4\pi}\sum_{n = 1}^\infty\frac{ik\left(ikr\right)^{n - 1}}{n!}dS\left(\mathbf{x}^\prime\right)dS\left(\mathbf{x}\right).
\end{equation}
The integral on the left is of the Laplace kernel, which this paper describes how to compute accurately.
The integral on the right is completely regular, and so can be computed via standard numerical means.
Because we assumed that $kr < 1$, the Taylor series is quickly convergent, and for any given accuracy, only a few terms need to be kept.

\section{Zero-Touch Case}
\label{zerocase}

In the zero-touch case, because the two triangles do not touch, the double surface integral is regular and can be computed via standard numerical or semi-analytical means.
However, in this section, we present an analytical method for computing the integral in this case.
This method uses spherical harmonics and multipole and local expansions and translations.
Similar methods were presented in \cite{of2005, pham2012} as part of solvers for problems in elastostatics.
However, we build on the methods described in these references by:
(1) adapting them to the kernels considered in Secs.\ \ref{background_sec} and \ref{double_surf_sec};
(2) computing the multipole expansion coefficients for a triangle exactly using Gaussian quadrature; and
(3) controlling the error by adaptively truncating the multipole expansions and/or subdividing the problem when necessary.

\subsection{Spherical Harmonics}
\label{sph_harm_sec}

\begin{figure}[t]
	\centering
	\includegraphics[scale=0.75]{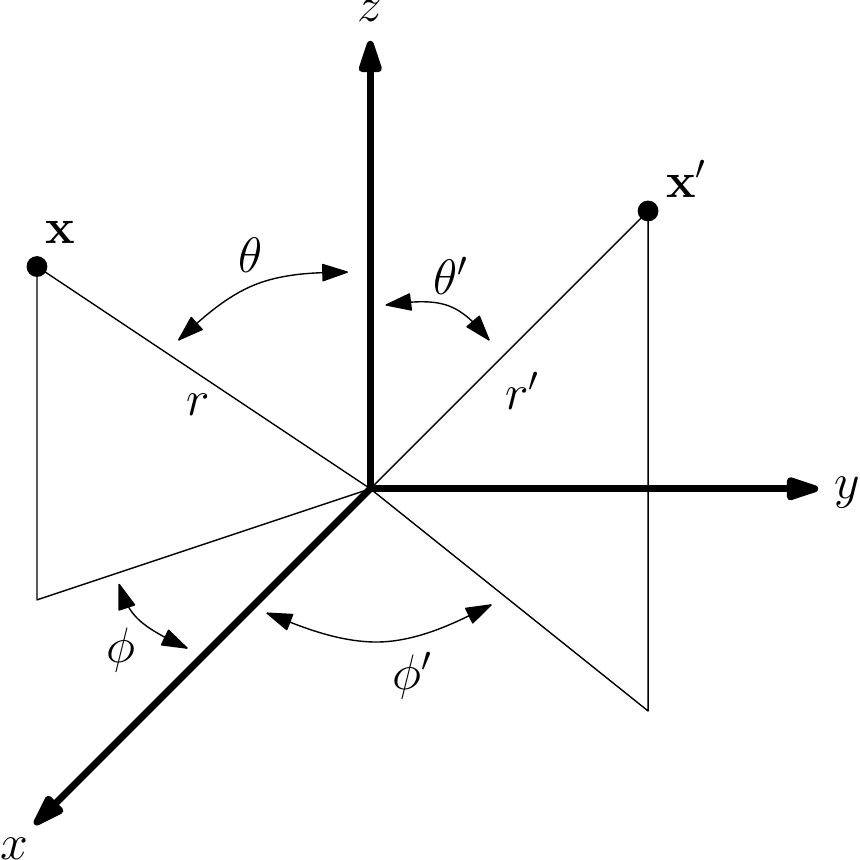}
	\caption{The spherical coordinate system.}
	\label{spherical_coords}
\end{figure}

The Laplace equation's Green's function can be expanded as
\begin{equation}
G\left(\mathbf{x} - \mathbf{x}^\prime\right) =
\frac{1}{4\pi\left|\mathbf{x} - \mathbf{x}^\prime\right|} =
\sum_{n = 0}^\infty
\sum_{m = -n}^n
\frac{1}{2n + 1}
\frac{r_<^n}{r_>^{n + 1}}
Y_n^{m\ast}\left(\theta^\prime, \phi^\prime\right)
Y_n^m\left(\theta, \phi\right),
\end{equation}
where $\left(r^\prime, \theta^\prime, \phi^\prime\right)$ and $\left(r, \theta, \phi\right)$ are the spherical coordinates of $\mathbf{x}^\prime$ and $\mathbf{x}$, respectively, $r_< = \min\left(r^\prime, r\right)$, and $r_> = \max\left(r^\prime, r\right)$ (see Fig.\ \ref{spherical_coords}).
The spherical harmonics are given by \cite{gumerov2005}
\begin{equation}
Y_n^m\left(\theta, \phi\right) = \left(-1\right)^m\left(\frac{2n + 1}{4\pi}\frac{\left(n - \left|m\right|\right)!}{\left(n + \left|m\right|\right)!}\right)^{1 / 2}P_n^{\left|m\right|}\left(\cos\left(\theta\right)\right)\exp\left(im\phi\right),
\end{equation}
where $P_n^m$ are the associated Legendre polynomials.

This expression can be used to build multipole and local expansions.
For example, suppose we want to compute the potential at $\mathbf{x}$ due to a point source at $\mathbf{x}^\prime$.
When $r > r^\prime$, we can build a multipole expansion:
\begin{equation}
\frac{1}{4\pi\left|\mathbf{x} - \mathbf{x}^\prime\right|} =
\sum_{n = 0}^\infty
\sum_{m = -n}^n
R_n^{m\ast}\left(r^\prime, \theta^\prime, \phi^\prime\right)
S_n^m\left(r, \theta, \phi\right).
\end{equation}
Here,
\begin{equation}
R_n^m\left(\mathbf{x}\right) = R_n^m\left(r, \theta, \phi\right) = \left(\frac{1}{2n + 1}\right)^{1 / 2}r^nY_n^m\left(\theta, \phi\right),
\end{equation}
\begin{equation}
S_n^m\left(\mathbf{x}\right) = S_n^m\left(r, \theta, \phi\right) = \left(\frac{1}{2n + 1}\right)^{1 / 2}\frac{1}{r^{n + 1}}Y_n^m\left(\theta, \phi\right)
\end{equation}
are the local and multipole expansion basis functions, respectively.
Instead of centering the expansion around the origin, we can center the expansion around $\mathbf{x}^\ast$.
When $\left|\mathbf{x} - \mathbf{x}^\ast\right| > \left|\mathbf{x}^\prime - \mathbf{x}^\ast\right|$,
\begin{equation}
\label{multipole_expansion}
\frac{1}{4\pi\left|\mathbf{x} - \mathbf{x}^\prime\right|} =
\sum_{n = 0}^\infty
\sum_{m = -n}^n
R_n^{m\ast}\left(\mathbf{x}^\prime - \mathbf{x}^\ast\right)
S_n^m\left(\mathbf{x} - \mathbf{x}^\ast\right).
\end{equation}
Likewise, when $\left|\mathbf{x} - \mathbf{x}^\ast\right| < \left|\mathbf{x}^\prime - \mathbf{x}^\ast\right|$, we can build a local expansion:
\begin{equation}
\frac{1}{4\pi\left|\mathbf{x} - \mathbf{x}^\prime\right|} =
\sum_{n = 0}^\infty
\sum_{m = -n}^n
S_n^{m\ast}\left(\mathbf{x}^\prime - \mathbf{x}^\ast\right)
R_n^m\left(\mathbf{x} - \mathbf{x}^\ast\right).
\end{equation}

These expressions can be used to build multipole and local expansions for arbitrary source distributions.
For example, let us build a multipole expansion for the source distribution, $\rho\left(\mathbf{x}^\prime\right)$, contained entirely inside an imaginary sphere of radius, $r^\ast$, centered around $\mathbf{x^\ast}$.
For a point, $\mathbf{x}$, outside the imaginary sphere, the potential due to this source distribution is given by
\begin{equation}
\Phi\left(\mathbf{x}\right) = \int_{\left|\mathbf{x}^\prime - \mathbf{x}^\ast\right| < r^\ast}\frac{\rho\left(\mathbf{x}^\prime\right)}{4\pi\left|\mathbf{x} - \mathbf{x}^\prime\right|}dV\left(\mathbf{x}^\prime\right).
\end{equation}
Plugging in Eq.\ (\ref{multipole_expansion}) and rearranging,
\begin{equation}
\Phi\left(\mathbf{x}\right) = \sum_{n = 0}^\infty
\sum_{m = -n}^na_n^m
S_n^m\left(\mathbf{x} - \mathbf{x}^\ast\right),
\end{equation}
where
\begin{equation}
a_n^m = \int_{\left|\mathbf{x}^\prime - \mathbf{x}^\ast\right| < r^\ast}\rho\left(\mathbf{x}^\prime\right)
R_n^{m\ast}\left(\mathbf{x}^\prime - \mathbf{x}^\ast\right)
dV\left(\mathbf{x}^\prime\right).
\end{equation}
We can build a local expansion using the same procedure.
Consider a different source distribution contained entirely outside the imaginary sphere.
For a point, $\mathbf{x}$, inside the imaginary sphere, the potential due to this source distribution is given by
\begin{equation}
\Phi\left(\mathbf{x}\right) = \sum_{n = 0}^\infty
\sum_{m = -n}^na_n^m
R_n^m\left(\mathbf{x} - \mathbf{x}^\ast\right),
\end{equation}
where
\begin{equation}
a_n^m = \int_{\left|\mathbf{x}^\prime - \mathbf{x}^\ast\right| > r^\ast}\rho\left(\mathbf{x}^\prime\right)
S_n^{m\ast}\left(\mathbf{x}^\prime - \mathbf{x}^\ast\right)
dV\left(\mathbf{x}^\prime\right).
\end{equation}

\subsection{Analytical Method}
\label{ana_method}

\begin{figure}[t]
	\centering
	\includegraphics[scale=0.75]{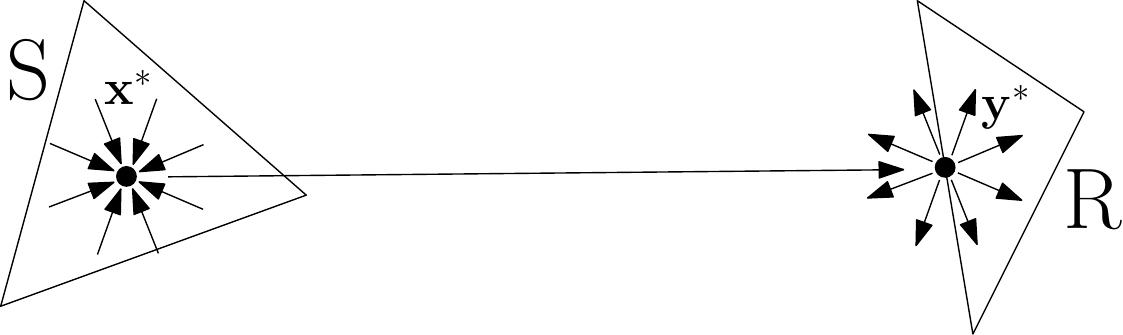}
	\caption{A diagram showing how to compute the double surface integral for two triangles that do not touch.}
	\label{zero_touch_problem}
\end{figure}

We want to compute the following double surface integral over a source triangle, $\text{S}$, and a receiver triangle, $\text{R}$:
\begin{equation}
I = \int_{\mathbf{x} \in \text{R}}
\left(\sigma_0 + \mathbf{p}\cdot\mathbf{x}\right)
\int_{\mathbf{x}^\prime \in \text{S}}\left(\sigma_0^\prime + \mathbf{p}^\prime\cdot\mathbf{x}^\prime\right)G\left(\mathbf{x} - \mathbf{x}^\prime\right)dS\left(\mathbf{x}^\prime\right)dS\left(\mathbf{x}\right).
\end{equation}

First, expand the Green's function as a multipole expansion:
\begin{equation}
\label{expression_expanded}
I = \int_{\mathbf{x} \in \text{R}}\left(\sigma_0 + \mathbf{p}\cdot\mathbf{x}\right)\int_{\mathbf{x}^\prime \in \text{S}}\left(\sigma_0^\prime + \mathbf{p}^\prime\cdot\mathbf{x}^\prime\right)\left(\sum_{n^\prime = 0}^\infty
\sum_{m^\prime = -n^\prime}^{n^\prime}
R_{n^\prime}^{m^\prime\ast}\left(\mathbf{x}^\prime - \mathbf{x}^\ast\right)
S_{n^\prime}^{m^\prime}\left(\mathbf{x} - \mathbf{x}^\ast\right)\right)dS\left(\mathbf{x}^\prime\right)dS\left(\mathbf{x}\right),
\end{equation}
where the expansion center, $\mathbf{x}^\ast$, is near the source triangle.
Ideally, $\mathbf{x}^\ast$ should be chosen so that the sphere centered around $\mathbf{x}^\ast$ and completely containing the source triangle can be made as small as possible.
There are actually only four possible choices of $\mathbf{x}^\ast$: the midpoints of the three edges of the source triangle and the center of the source triangle's circumsphere.
The one corresponding to the smallest sphere that completely contains the source triangle is chosen.

Second, rearrange Eq.\ (\ref{expression_expanded}) by moving the double sum and the $S_{n^\prime}^{m^\prime}$ outside the inside integral:
\begin{equation}
I = \int_{\mathbf{x} \in \text{R}}\left(\sigma_0 + \mathbf{p}\cdot\mathbf{x}\right)\left(\sum_{n^\prime = 0}^\infty
\sum_{m^\prime = -n^\prime}^{n^\prime}\left(\int_{\mathbf{x}^\prime \in \text{S}}\left(\sigma_0^\prime + \mathbf{p}^\prime\cdot\mathbf{x}^\prime\right)
R_{n^\prime}^{m^\prime\ast}\left(\mathbf{x}^\prime - \mathbf{x}^\ast\right)
dS\left(\mathbf{x}^\prime\right)\right)S_{n^\prime}^{m^\prime}\left(\mathbf{x} - \mathbf{x}^\ast\right)\right)dS\left(\mathbf{x}\right).
\end{equation}
The integral over the source triangle now computes the expansion coefficients for the multipole expansion that represents the potential due to the linear source distribution over the source triangle:
\begin{equation}
I = \int_{\mathbf{x} \in \text{R}}\left(\sigma_0 + \mathbf{p}\cdot\mathbf{x}\right)\left(\sum_{n^\prime = 0}^\infty
\sum_{m^\prime = -n^\prime}^{n^\prime}a_{n^\prime}^{m^\prime}S_{n^\prime}^{m^\prime}\left(\mathbf{x} - \mathbf{x}^\ast\right)\right)dS\left(\mathbf{x}\right),
\end{equation}
where
\begin{equation}
\label{anm_expression}
a_{n^\prime}^{m^\prime} = \int_{\mathbf{x}^\prime \in \text{S}}\left(\sigma_0^\prime + \mathbf{p}^\prime\cdot\mathbf{x}^\prime\right)
R_{n^\prime}^{m^\prime\ast}\left(\mathbf{x}^\prime - \mathbf{x}^\ast\right)
dS\left(\mathbf{x}^\prime\right).
\end{equation}
There are many analytical expressions available for computing the integral in Eq.\ (\ref{anm_expression}), including those in \cite{newman1986, nabors1993}.
In \cite{barrett2014}, the authors presented a recursive algorithm for computing the expansion coefficients: only $a_{n^\prime}^{m^\prime}$ for lower order and degree need to be computed explicity; the others can be computed recursively from them.
However, we make the following observation: $R_{n^\prime}^{m^\prime}$ is polynomial, so the integrand in Eq.\ (\ref{anm_expression}) is polynomial.
Thus, the integral can be computed exactly via Gaussian quadrature.
A similar approach was used in \cite{gumerov2013, gumerov2014}, although the integration domains in these references were lines and boxes, not triangles.
We use the techniques given in \cite{deng2010} for performing Gaussian quadrature over the triangles.

Third, translate the multipole expansion centered around $\mathbf{x}^\ast$ to a local expansion centered around $\mathbf{y}^\ast$, where $\mathbf{y}^\ast$ is near the receiver triangle \cite{gumerov2005}:
\begin{equation}
\label{local_expansion}
I = \int_{\mathbf{x} \in \text{R}}\left(\sigma_0 + \mathbf{p}\cdot\mathbf{x}\right)\left(\sum_{n = 0}^\infty
\sum_{m = -n}^nb_n^mR_n^m\left(\mathbf{x} - \mathbf{y}^\ast\right)\right)dS\left(\mathbf{x}\right).
\end{equation}
Ideally, $\mathbf{y}^\ast$ should be chosen so that the sphere centered around $\mathbf{y}^\ast$ and completely containing the receiver triangle can be made as small as possible.
The procedure for choosing $\mathbf{y}^\ast$ is the same as for choosing $\mathbf{x}^\ast$.

Fourth, rearrange Eq.\ (\ref{local_expansion}) by moving the double sum and the expansion coefficients outside the integral:
\begin{equation}
I = \sum_{n = 0}^\infty
\sum_{m = -n}^n
b_n^m
\left(
\int_{\mathbf{x} \in \text{R}}
\left(\sigma_0 + \mathbf{p}\cdot\mathbf{x}\right)
R_n^m\left(\mathbf{x} - \mathbf{y}^\ast\right)
dS\left(\mathbf{x}\right)
\right),
\end{equation}
\begin{equation}
\label{last_expression}
I = \sum_{n = 0}^\infty
\sum_{m = -n}^n
b_n^m
c_n^m,
\end{equation}
where
\begin{equation}
\label{cnm_expression}
c_n^m =
\int_{\mathbf{x} \in \text{R}}
\left(\sigma_0 + \mathbf{p}\cdot\mathbf{x}\right)
R_n^m\left(\mathbf{x} - \mathbf{y}^\ast\right)
dS\left(\mathbf{x}\right).
\end{equation}
Like before, since $R_n^m$ is polynomial, the integral in Eq.\ (\ref{cnm_expression}) can be computed exactly via Gaussian quadrature.

\subsection{Error Control}

\begin{figure}[t]
	\centering
	\includegraphics[scale=0.75]{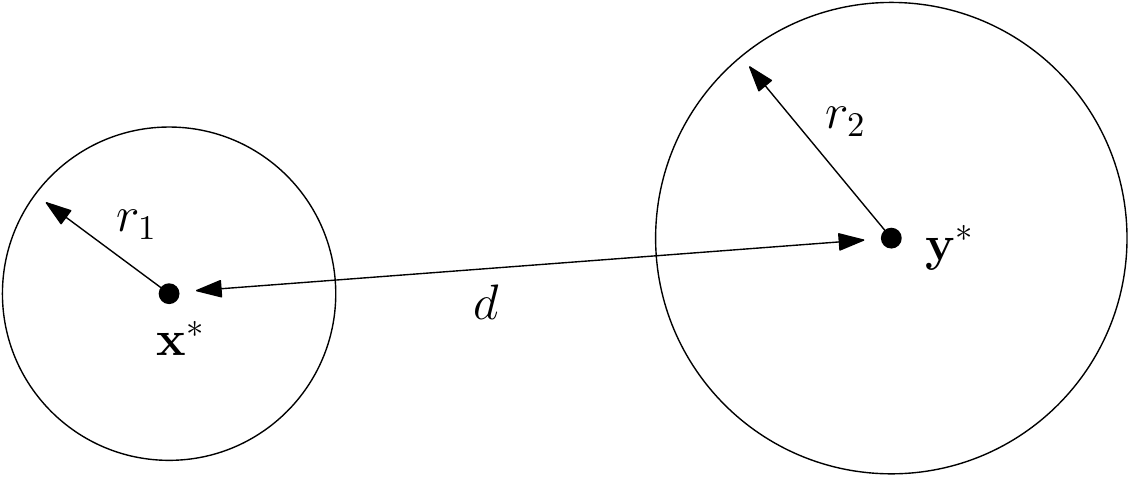}
	\caption{The radii and distances used to compute the geometric factor, $\eta$, given in Eq.\ (\ref{eta_def_eq}).}
	\label{eta_def}
\end{figure}

\begin{figure}[t]
	\centering
	\includegraphics[scale=0.5]{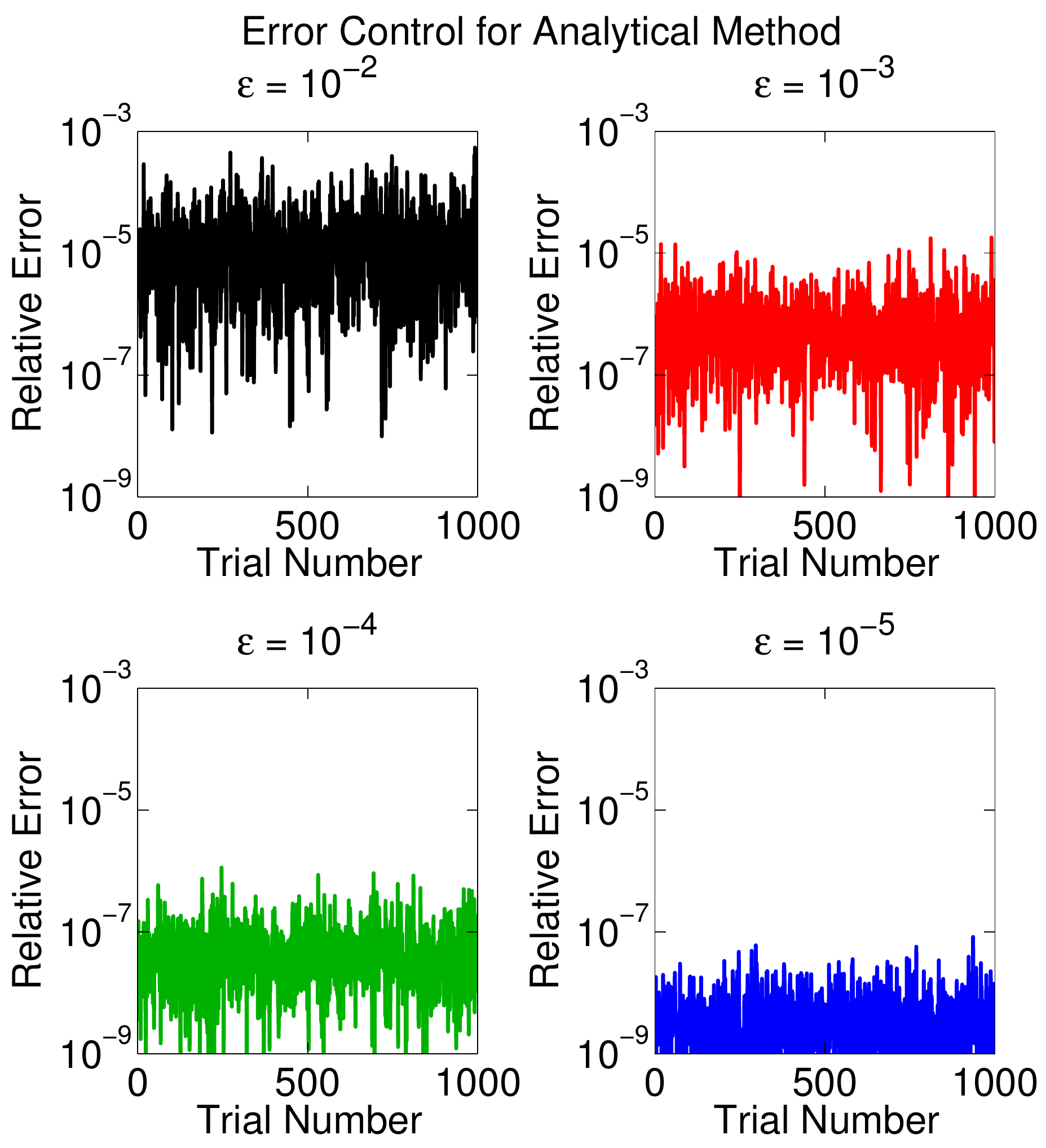}
	\caption{This graph shows the actual relative errors when computing the double surface integrals over 1,000 randomly generated pairs of triangles using the analytical method for several different error bounds.  In all cases, these errors were below the error bounds.}
	\label{validate_errors_linear}
\end{figure}

The expressions derived in Secs.\ \ref{sph_harm_sec} and \ref{ana_method} involving multipole and local expansions must be truncated so that they can be implemented in code.
For example, the last expression given in Sec.\ \ref{ana_method} becomes
\begin{equation}
\label{truncated_I_p}
I \approx I_p = \sum_{n = 0}^{p - 1}
\sum_{m = -n}^n
b_n^m
c_n^m,
\end{equation}
where only $p^2$ terms have been kept.
Obviously, the expression is exact as $p \rightarrow \infty$, but there are truncation errors when $p < \infty$.
These errors come from two sources:
(1) the construction of the multipole expansion at the source triangle; and
(2) the translation of the multipole expansion to a local expansion at the receiver triangle.
Luckily, these expansions converge rapidly in the case of the Laplace equation, so these errors can be precisely controlled by picking an appropriate value of $p$.

Theoretical bounds as a function of $p$ for these errors have been derived by many authors over the years.
A good overview is given in \cite{gumerov2008}.
The relative error is bounded by:
\begin{equation}
\label{error_bound_equation}
\varepsilon = \left|\frac{I_p - I}{I}\right| \le A\eta^p,
\end{equation}
where the error constant, $A$, depends on the problem being solved, and
\begin{equation}
\label{eta_def_eq}
\eta = \frac{\max\left(r_1, r_2\right)}{d - \min\left(r_1, r_2\right)},
\end{equation}
where $r_1$ is the radius of the multipole expansion's bounding sphere, $r_2$ is the radius of the local expansion's bounding sphere, and $d = \left|\mathbf{y}^\ast - \mathbf{x}^\ast\right|$ (see Fig.\ \ref{eta_def}).
Using Eq.\ (\ref{error_bound_equation}), we can easily pick a truncation number that gives us a desired accuracy:
\begin{equation}
p = \left\lceil\left.\log\left(\cfrac{\varepsilon}{A}\right)\middle/\log\left(\eta\right)\right.\right\rceil,
\end{equation}
where $\left\lceil{}x\right\rceil$ is the ceiling of $x$.

There are two issues.
First, we may not always be able to use as high a value of $p$ as we want.
Since the number of terms in Eq.\ (\ref{truncated_I_p}) grows as $p^2$, the memory and computational costs grow as $p^2$ as well.
Therefore, for practical reasons, $p$ must be capped.
In the event that $p$ needs to be higher than this cap, we divide the larger triangle and recurse.
We do the same thing in the event that the two bounding spheres overlap.

Second, while values of $A$ have been derived for special cases, such as point sources, no such values have been derived for the case of triangles.
Instead of attempting to derive such a value analytically, we computed one experimentally.
The experiment worked as follows.
We generated 10,000 pairs of randomly placed triangles, where, in each pair, the two triangles did not touch.
For each pair, we computed $I$ using a semi-analytical method (using a high-order Gaussian quadrature for the outside integral) and $I_p$ for $p = 1, 2, \ldots, 10$.
We chose a value for $A$ so that the experimental data satisfied Eq.\ (\ref{error_bound_equation}).
We ran this experiment several times, and found $A$ to be between $0.5$ and $0.6$, so we set $A = 0.6$.

We ran a second experiment to verify that $A = 0.6$ works well.
The experiment worked as follows.
We generated 1,000 pairs of triangles like before.
For each pair, we computed $I$ using the analytical method for several different choices of $\varepsilon$: $10^{-2}$, $10^{-3}$, $10^{-4}$, and $10^{-5}$.
Again, we used the value of $I$ returned by the semi-analytical method as the reference value.
In Fig.\ \ref{validate_errors_linear}, the actual errors (along the $y$ axis) for each pair of triangles (along the $x$ axis) are plotted for each choice of $\epsilon$ (the different colored curves).
In all cases, the realized errors are below the desired error bounds.

\section{One-, Two-, and Three-Touch Cases}
\label{onetwothreecase}

In the one-, two-, and three-touch cases, because the two triangles touch, the double surface integrals are not regular, so they cannot be computed using standard numerical or semi-analytical means, or even the analytical method described in Sec.\ \ref{ana_method}.
In this section, we present a method for dealing with these integrals.
The approach relies on several scaling properties of the integrals and the kernels being integrated.
The method works in roughly the following way:
(1) the integral is broken up into several smaller integrals;
(2) some of these integrals are written in terms of the original integral using some simple analysis; and
(3) the terms are rearranged to yield an expression for the original integral that only requires computing regular integrals explicitly (all other integrals are computed implicitly).

\subsection{Preliminaries: Scaling Results}

\begin{figure}[t]
	\centering
	\includegraphics[scale=0.75]{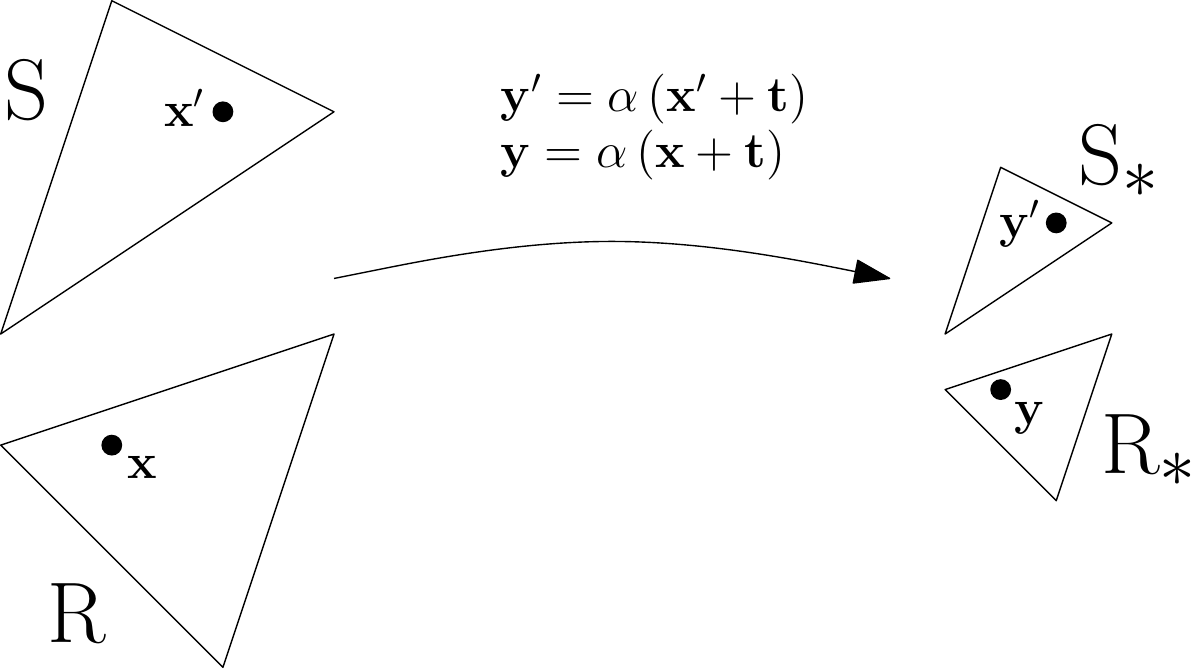}
	\caption{A diagram showing the process for transforming the pair of triangles, $\text{S}$ and $\text{R}$, into the pair of triangles, $\text{S}_\ast$ and $\text{R}_\ast$.}
	\label{preliminaries}
\end{figure}

We want to compute the following double surface integral over a source triangle, $\text{S}$, and a receiver triangle, $\text{R}$:
\begin{equation}
I = \int_{\mathbf{x} \in \text{R}}\left(\sigma_0 + \mathbf{p}\cdot\mathbf{x}\right)\int_{\mathbf{x}^\prime \in \text{S}}\left(\sigma_0^\prime + \mathbf{p}^\prime\cdot\mathbf{x}^\prime\right)F\left(\mathbf{x} - \mathbf{x}^\prime\right)dS\left(\mathbf{x}^\prime\right)dS\left(\mathbf{x}\right),
\end{equation}
where $F\left(\mathbf{r}\right)$ is any kernel that has the following scaling property:
\begin{equation}
\label{scalingproperty}
F\left(\alpha\mathbf{r}\right) = s\left(\alpha\right)F\left(\mathbf{r}\right).
\end{equation}
To begin, let us break the integral into four smaller integrals by expanding the product of the two linear functions:
\begin{equation}
I = \int_{\mathbf{x} \in \text{R}}\int_{\mathbf{x}^\prime \in \text{S}}\left(
\sigma_0\sigma_0^\prime +
\sigma_0\left(\mathbf{p}^\prime\cdot\mathbf{x}^\prime\right) +
\left(\mathbf{p}\cdot\mathbf{x}\right)\sigma_0^\prime +
\left(\mathbf{p}\cdot\mathbf{x}\right)\left(\mathbf{p}^\prime\cdot\mathbf{x}^\prime\right)
\right)
F\left(\mathbf{x} - \mathbf{x}^\prime\right)dS\left(\mathbf{x}^\prime\right)dS\left(\mathbf{x}\right),
\end{equation}
\begin{equation}
I = \sigma_0\sigma_0^\prime{}I^1 + \sigma_0I^{\mathbf{p}^\prime} + \sigma_0^\prime{}I^{\mathbf{p}} + I^{\mathbf{p}^\prime\mathbf{p}},
\end{equation}
where
\begin{equation}
I^1 = \int_{\mathbf{x} \in \text{R}}\int_{\mathbf{x}^\prime \in \text{S}}F\left(\mathbf{x} - \mathbf{x}^\prime\right)dS\left(\mathbf{x}^\prime\right)dS\left(\mathbf{x}\right),
\end{equation}
\begin{equation}
I^{\mathbf{p}^\prime} = \int_{\mathbf{x} \in \text{R}}\int_{\mathbf{x}^\prime \in \text{S}}\left(\mathbf{p}^\prime\cdot\mathbf{x}^\prime\right)F\left(\mathbf{x} - \mathbf{x}^\prime\right)dS\left(\mathbf{x}^\prime\right)dS\left(\mathbf{x}\right),
\end{equation}
\begin{equation}
I^{\mathbf{p}} = \int_{\mathbf{x} \in \text{R}}\int_{\mathbf{x}^\prime \in \text{S}}\left(\mathbf{p}\cdot\mathbf{x}\right)F\left(\mathbf{x} - \mathbf{x}^\prime\right)dS\left(\mathbf{x}^\prime\right)dS\left(\mathbf{x}\right),
\end{equation}
\begin{equation}
I^{\mathbf{p}^\prime\mathbf{p}} = \int_{\mathbf{x} \in \text{R}}\int_{\mathbf{x}^\prime \in \text{S}}\left(\mathbf{p}^\prime\cdot\mathbf{x}^\prime\right)\left(\mathbf{p}\cdot\mathbf{x}\right)F\left(\mathbf{x} - \mathbf{x}^\prime\right)dS\left(\mathbf{x}^\prime\right)dS\left(\mathbf{x}\right).
\end{equation}

Consider two different triangles, $\text{S}_\ast$ and $\text{R}_\ast$, which, taken together, are scaled and translated versions of $\text{S}$ and $\text{R}$, also taken together (see Fig.\ \ref{preliminaries}).
In other words, given a pair of points, $\mathbf{x}^\prime, \mathbf{x} \in \text{S} \cup \text{R}$, there is a corresponding pair of points, $\mathbf{y}^\prime, \mathbf{y} \in \text{S}_\ast \cup \text{R}_\ast$, given by
\begin{equation}
\label{changeofvars}
\mathbf{y}^\prime = \alpha\left(\mathbf{x}^\prime + \mathbf{t}\right),\quad\mathbf{y} = \alpha\left(\mathbf{x} + \mathbf{t}\right).
\end{equation}
Suppose we want to compute the same integral as before, except over $\text{S}_\ast$ and $\text{R}_\ast$:
\begin{equation}
I_\ast = \int_{\mathbf{y} \in \text{R}_\ast}\int_{\mathbf{y}^\prime \in \text{S}_\ast}\left(\sigma_0 + \mathbf{p}\cdot\mathbf{y}\right)\left(\sigma_0^\prime + \mathbf{p}^\prime\cdot\mathbf{y}^\prime\right)F\left(\mathbf{y} - \mathbf{y}^\prime\right)dS\left(\mathbf{y}^\prime\right)dS\left(\mathbf{y}\right).
\end{equation}
The integrand is exactly the same as before.
The only thing we are changing is the integration domain from $\text{S}\times\text{R}$ to $\text{S}_\ast\times\text{R}_\ast$.
Like before, let us break the integral into four smaller integrals:
\begin{equation}
I_\ast = \sigma_0\sigma_0^\prime{}I_\ast^1 + \sigma_0I_\ast^{\mathbf{p}^\prime} + \sigma_0^\prime{}I_\ast^{\mathbf{p}} + I_\ast^{\mathbf{p}^\prime\mathbf{p}},
\end{equation}
where
\begin{equation}
I_\ast^1 = \int_{\mathbf{y} \in \text{R}_\ast}\int_{\mathbf{y}^\prime \in \text{S}_\ast}F\left(\mathbf{y} - \mathbf{y}^\prime\right)dS\left(\mathbf{y}^\prime\right)dS\left(\mathbf{y}\right),
\end{equation}
\begin{equation}
I_\ast^{\mathbf{p}^\prime} = \int_{\mathbf{y} \in \text{R}_\ast}\int_{\mathbf{y}^\prime \in \text{S}_\ast}\left(\mathbf{p}^\prime\cdot\mathbf{y}^\prime\right)F\left(\mathbf{y} - \mathbf{y}^\prime\right)dS\left(\mathbf{y}^\prime\right)dS\left(\mathbf{y}\right),
\end{equation}
\begin{equation}
I_\ast^{\mathbf{p}} = \int_{\mathbf{y} \in \text{R}_\ast}\int_{\mathbf{y}^\prime \in \text{S}_\ast}\left(\mathbf{p}\cdot\mathbf{y}\right)F\left(\mathbf{y} - \mathbf{y}^\prime\right)dS\left(\mathbf{y}^\prime\right)dS\left(\mathbf{y}\right),
\end{equation}
\begin{equation}
I_\ast^{\mathbf{p}^\prime\mathbf{p}} = \int_{\mathbf{y} \in \text{R}_\ast}\int_{\mathbf{y}^\prime \in \text{S}_\ast}\left(\mathbf{p}^\prime\cdot\mathbf{y}^\prime\right)\left(\mathbf{p}\cdot\mathbf{y}\right)F\left(\mathbf{y} - \mathbf{y}^\prime\right)dS\left(\mathbf{y}^\prime\right)dS\left(\mathbf{y}\right).
\end{equation}

\begin{theorem}
\label{scaling_thm}
The four integrals over $\text{S}_\ast$ and $\text{R}_\ast$, $I_\ast^1$, $I_\ast^{\mathbf{p}^\prime}$, $I_\ast^{\mathbf{p}}$, and $I_\ast^{\mathbf{p}^\prime\mathbf{p}}$, can be expressed in terms of the four integrals over $\text{S}$ and $\text{R}$, $I^1$, $I^{\mathbf{p}^\prime}$, $I^{\mathbf{p}}$, and $I^{\mathbf{p}^\prime\mathbf{p}}$:
\begin{equation}
\label{I1ast}
I_\ast^1 = s\left(\alpha\right)\alpha^4I^1,
\end{equation}
\begin{equation}
\label{Ipprimeast}
I_\ast^{\mathbf{p}^\prime} = s\left(\alpha\right)\alpha^5\left(I^{\mathbf{p}^\prime} + \left(\mathbf{p}^\prime\cdot\mathbf{t}\right)I^1\right),
\end{equation}
\begin{equation}
\label{Ipast}
I_\ast^{\mathbf{p}} = s\left(\alpha\right)\alpha^5\left(I^{\mathbf{p}} + \left(\mathbf{p}\cdot\mathbf{t}\right)I^1\right),
\end{equation}
\begin{equation}
\label{Ipprimepast}
I_\ast^{\mathbf{p}^\prime\mathbf{p}} = s\left(\alpha\right)\alpha^6\left(
I^{\mathbf{p}^\prime\mathbf{p}} +
\left(\mathbf{p}\cdot\mathbf{t}\right)I^{\mathbf{p}^\prime} +
\left(\mathbf{p}^\prime\cdot\mathbf{t}\right)I^{\mathbf{p}} +
\left(\mathbf{p}\cdot\mathbf{t}\right)\left(\mathbf{p}^\prime\cdot\mathbf{t}\right)I^1
\right).
\end{equation}
\end{theorem}

\begin{proof}
To prove Eqs.\ (\ref{I1ast}) - (\ref{Ipprimepast}) in Theorem \ref{scaling_thm}, we are going to use the following procedure: (1) make the change of variables in Eq.\ (\ref{changeofvars}); (2) use the scaling property of $F\left(\mathbf{r}\right)$ in Eq.\ (\ref{scalingproperty}); and (3) break the resulting integral into one or more integrals that are equal to the original four integrals, $I^1$, $I^{\mathbf{p}^\prime}$, $I^{\mathbf{p}}$, and $I^{\mathbf{p}^\prime\mathbf{p}}$.
First, let us prove Eq.\ (\ref{I1ast}):
\begin{equation}
I_\ast^1 = \int_{\mathbf{y} \in \text{R}_\ast}\int_{\mathbf{y}^\prime \in \text{S}_\ast}F\left(\mathbf{y} - \mathbf{y}^\prime\right)dS\left(\mathbf{y}^\prime\right)dS\left(\mathbf{y}\right),
\end{equation}
\begin{equation}
I_\ast^1 = \int_{\mathbf{x} \in \text{R}}\int_{\mathbf{x}^\prime \in \text{S}}F\left(\alpha\left(\mathbf{x} + \mathbf{t}\right) - \alpha\left(\mathbf{x}^\prime + \mathbf{t}\right)\right)dS\left(\alpha\left(\mathbf{x}^\prime + \mathbf{t}\right)\right)dS\left(\alpha\left(\mathbf{x} + \mathbf{t}\right)\right),
\end{equation}
\begin{equation}
I_\ast^1 = \alpha^4\int_{\mathbf{x} \in \text{R}}\int_{\mathbf{x}^\prime \in \text{S}}F\left(\alpha\left(\mathbf{x} - \mathbf{x}^\prime\right)\right)dS\left(\mathbf{x}^\prime\right)dS\left(\mathbf{x}\right),
\end{equation}
\begin{equation}
I_\ast^1 = s\left(\alpha\right)\alpha^4\int_{\mathbf{x} \in \text{R}}\int_{\mathbf{x}^\prime \in \text{S}}F\left(\mathbf{x} - \mathbf{x}^\prime\right)dS\left(\mathbf{x}^\prime\right)dS\left(\mathbf{x}\right),
\end{equation}
\begin{equation}
I_\ast^1 = s\left(\alpha\right)\alpha^4I^1.
\end{equation}
Second, let us prove Eq.\ (\ref{Ipprimeast}):
\begin{equation}
I_\ast^{\mathbf{p}^\prime} = \int_{\mathbf{y} \in \text{R}_\ast}\int_{\mathbf{y}^\prime \in \text{S}_\ast}\left(\mathbf{p}^\prime\cdot\mathbf{y}^\prime\right)F\left(\mathbf{y} - \mathbf{y}^\prime\right)dS\left(\mathbf{y}^\prime\right)dS\left(\mathbf{y}\right),
\end{equation}
\begin{equation}
I_\ast^{\mathbf{p}^\prime} = s\left(\alpha\right)\alpha^5\int_{\mathbf{x} \in \text{R}}\int_{\mathbf{x}^\prime \in \text{S}}\left(\mathbf{p}^\prime\cdot\left(\mathbf{x}^\prime + \mathbf{t}\right)\right)F\left(\mathbf{x} - \mathbf{x}^\prime\right)dS\left(\mathbf{x}^\prime\right)dS\left(\mathbf{x}\right),
\end{equation}
\begin{equation}
I_\ast^{\mathbf{p}^\prime} = s\left(\alpha\right)\alpha^5\left(I^{\mathbf{p}^\prime} + \left(\mathbf{p}^\prime\cdot\mathbf{t}\right)I^1\right).
\end{equation}
Third, let us prove Eq.\ (\ref{Ipast}):
\begin{equation}
I_\ast^{\mathbf{p}} = \int_{\mathbf{y} \in \text{R}_\ast}\int_{\mathbf{y}^\prime \in \text{S}_\ast}\left(\mathbf{p}\cdot\mathbf{y}\right)F\left(\mathbf{y} - \mathbf{y}^\prime\right)dS\left(\mathbf{y}^\prime\right)dS\left(\mathbf{y}\right).
\end{equation}
The analysis is exactly the same as for $I^{\mathbf{p}^\prime}$, so
\begin{equation}
I_\ast^{\mathbf{p}} = s\left(\alpha\right)\alpha^5\left(I^{\mathbf{p}} + \left(\mathbf{p}\cdot\mathbf{t}\right)I^1\right).
\end{equation}
Fourth, let us prove Eq.\ (\ref{Ipprimepast}):
\begin{equation}
I_\ast^{\mathbf{p}^\prime\mathbf{p}} = \int_{\mathbf{y} \in \text{R}_\ast}\int_{\mathbf{y}^\prime \in \text{S}_\ast}\left(\mathbf{p}\cdot\mathbf{y}\right)\left(\mathbf{p}^\prime\cdot\mathbf{y}^\prime\right)F\left(\mathbf{y} - \mathbf{y}^\prime\right)dS\left(\mathbf{y}^\prime\right)dS\left(\mathbf{y}\right),
\end{equation}
\begin{equation}
I_\ast^{\mathbf{p}^\prime\mathbf{p}} = s\left(\alpha\right)\alpha^6\int_{\mathbf{x} \in \text{R}}\int_{\mathbf{x}^\prime \in \text{S}}\left(\mathbf{p}\cdot\left(\mathbf{x} + \mathbf{t}\right)\right)\left(\mathbf{p}^\prime\cdot\left(\mathbf{x}^\prime + \mathbf{t}\right)\right)F\left(\mathbf{x} - \mathbf{x}^\prime\right)dS\left(\mathbf{x}^\prime\right)dS\left(\mathbf{x}\right),
\end{equation}
\begin{equation}
I_\ast^{\mathbf{p}^\prime\mathbf{p}} = s\left(\alpha\right)\alpha^6\left(
I^{\mathbf{p}^\prime\mathbf{p}} +
\left(\mathbf{p}\cdot\mathbf{t}\right)I^{\mathbf{p}^\prime} +
\left(\mathbf{p}^\prime\cdot\mathbf{t}\right)I^{\mathbf{p}} +
\left(\mathbf{p}\cdot\mathbf{t}\right)\left(\mathbf{p}^\prime\cdot\mathbf{t}\right)I^1
\right).
\end{equation}
At this point, we have proved Theorem 1.
\end{proof}

\subsection{One-Touch Case}
\label{onetouch}

\begin{figure}[t]
	\centering
	\includegraphics[scale=0.75]{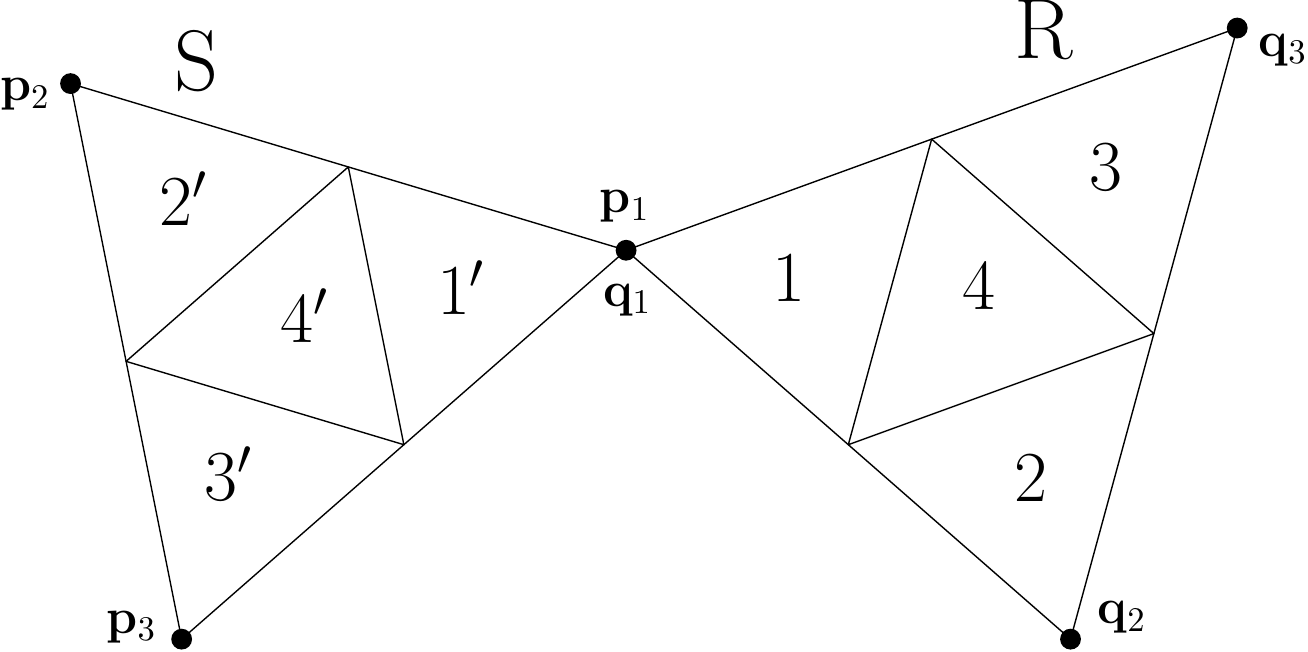}
	\caption{A diagram showing how to compute the double surface integral for two triangles that share a vertex.
The source triangle, $\text{S}$, is formed by the three vertices, $\mathbf{p}_1$, $\mathbf{p}_2$, $\mathbf{p}_3$, and the receiver triangle, $\text{R}$, is formed by the three vertices, $\mathbf{q}_1$, $\mathbf{q}_2$, and $\mathbf{q}_3$.
The two triangles share a vertex: $\mathbf{p}_1 = \mathbf{q}_1$.}
	\label{one_touch_problem}
\end{figure}

Consider the one-touch case in Fig.\ \ref{one_touch_problem}.
Without loss of generality, assume the vertex that the two triangles share is located at the origin (i.e., $\mathbf{p}_1 = \mathbf{0}$).
The fact that they share a vertex complicates the double surface integral over them.
We solve this problem by dividing each triangle into four smaller, congruent triangles: the source triangle, $\text{S}$, is divided into $1^\prime$, $2^\prime$, $3^\prime$, and $4^\prime$, and the receiver triangle, $\text{R}$, is divided into $1$, $2$, $3$, and $4$.
Then, we break the integral into seven smaller integrals over these triangles.

Let us look at $I^1$:
\begin{equation}
\label{one_touch_1_everything}
I^1 = I_{1^\prime1}^1 + I_{1^\prime2}^1 + I_{1^\prime3}^1 + I_{1^\prime4}^1 + I_{2^\prime\text{R}}^1 + I_{3^\prime\text{R}}^1 + I_{4^\prime\text{R}}^1.
\end{equation}
The subscripts denote the surfaces of integration.
For example, $I_{1^\prime1}^1$ is the integral over $1^\prime$ and $1$.
Similarly, $I_{2^\prime\text{R}}^1$ is the integral over $2^\prime$ and $\text{R}$.
The $I^1$ without a subscript is the original integral over $\text{S}$ and $\text{R}$.
In Eq.\ (\ref{one_touch_1_everything}), the six integrals, $I_{1^\prime2}^1$, $I_{1^\prime3}^1$, $I_{1^\prime4}^1$, $I_{2^\prime\text{R}}^1$, $I_{3^\prime\text{R}}^1$, and $I_{4^\prime\text{R}}^1$, correspond to pairs of triangles that do not touch.
Because we know how to compute them, let us combine them into a single integral:
\begin{equation}
\label{one_touch_1}
I^1 = I_{1^\prime1}^1 + I_{\text{remainder}}^1,
\end{equation}
where
\begin{equation}
I_{\text{remainder}}^1 = I_{1^\prime2}^1 + I_{1^\prime3}^1 + I_{1^\prime4}^1 + I_{2^\prime\text{R}}^1 + I_{3^\prime\text{R}}^1 + I_{4^\prime\text{R}}^1.
\end{equation}
The integral, $I_{1^\prime1}^1$, however, has the same problem as the original integral: $1^\prime$ and $1$ share a vertex.
Fortunately, $1^\prime$ and $1$ are scaled and translated versions of $\text{S}$ and $\text{R}$.
We derived earlier that
\begin{equation}
I_\ast^1 = s\left(\alpha\right)\alpha^4I^1.
\end{equation}
Since the pair of triangles, $1^\prime$ and $1$, is $1 / 2$ the size of the original pair, $\text{S}$ and $\text{R}$, $\alpha = 1 / 2$, so
\begin{equation}
I_{1^\prime1}^1 = s\left(\frac{1}{2}\right)\frac{1}{16}I^1.
\end{equation}
Plugging this into Eq.\ (\ref{one_touch_1}) and rearranging,
\begin{equation}
I^1 = \left(1 - s\left(\frac{1}{2}\right)\frac{1}{16}\right)^{-1}I_{\text{remainder}}^1.
\end{equation}
The integral, $I_{1^\prime1}^1$, is computed implicitly in this expression.

Now, let us look at $I^{\mathbf{p}^\prime}$:
\begin{equation}
I^{\mathbf{p}^\prime} = I_{1^\prime1}^{\mathbf{p}^\prime} + I_{1^\prime2}^{\mathbf{p}^\prime} + I_{1^\prime3}^{\mathbf{p}^\prime} + I_{1^\prime4}^{\mathbf{p}^\prime} + I_{2^\prime\text{R}}^{\mathbf{p}^\prime} + I_{3^\prime\text{R}}^{\mathbf{p}^\prime} + I_{4^\prime\text{R}}^{\mathbf{p}^\prime}.
\end{equation}
Again, the six integrals other than $I_{1^\prime1}^{\mathbf{p}^\prime}$ are all regular, so let us combine them into a single integral:
\begin{equation}
\label{one_touch_pprime}
I^{\mathbf{p}^\prime} = I_{1^\prime1}^{\mathbf{p}^\prime} + I_{\text{remainder}}^{\mathbf{p}^\prime}.
\end{equation}
This leaves $I_{1^\prime1}^{\mathbf{p}^\prime}$.
We derived earlier that
\begin{equation}
I_\ast^{\mathbf{p}^\prime} = s\left(\alpha\right)\alpha^5\left(I^{\mathbf{p}^\prime} + \left(\mathbf{p}^\prime\cdot\mathbf{t}\right)I^1\right).
\end{equation}
Like before, $\alpha = 1 / 2$, but we still need to determine $\mathbf{t}$.
When $\alpha = 1 / 2$, $\mathbf{t}$ is the point that does not change during the scaling and translation.
This is simply the vertex that the two triangles share, so $\mathbf{t} = \mathbf{p}_1$.
Because we assumed that $\mathbf{p}_1 = \mathbf{0}$, $\mathbf{t} = \mathbf{0}$ as well, so
\begin{equation}
I_{11}^{\mathbf{p}^\prime} = s\left(\frac{1}{2}\right)\frac{1}{32}I^{\mathbf{p}^\prime}.
\end{equation}
Plugging this into Eq.\ (\ref{one_touch_pprime}) and rearranging,
\begin{equation}
I^{\mathbf{p}^\prime} = \left(1 - s\left(\frac{1}{2}\right)\frac{1}{32}\right)^{-1}I_{\text{remainder}}^{\mathbf{p}^\prime}.
\end{equation}

These same procedures can be used to compute $I^{\mathbf{p}}$ and $I^{\mathbf{p}^\prime\mathbf{p}}$:
\begin{equation}
I^{\mathbf{p}} = \left(1 - s\left(\frac{1}{2}\right)\frac{1}{32}\right)^{-1}I_{\text{remainder}}^{\mathbf{p}},
\end{equation}
\begin{equation}
I^{\mathbf{p}^\prime\mathbf{p}} = \left(1 - s\left(\frac{1}{2}\right)\frac{1}{64}\right)^{-1}I_{\text{remainder}}^{\mathbf{p}^\prime\mathbf{p}}.
\end{equation}

\subsection{Two-Touch Case}
\label{twotouch}

\begin{figure}[t]
	\centering
	\includegraphics[scale=0.75]{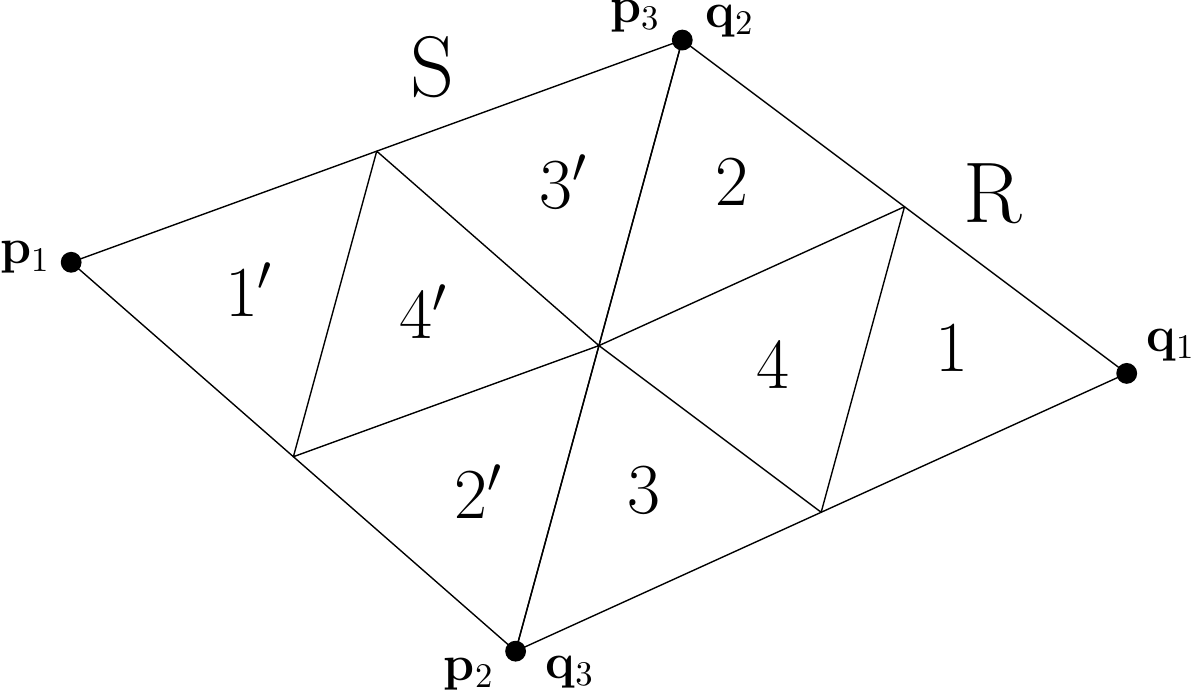}
	\caption{A diagram showing how to compute the double surface integral for two triangles that share an edge.
The source triangle, $\text{S}$, is formed by the three vertices, $\mathbf{p}_1$, $\mathbf{p}_2$, $\mathbf{p}_3$, and the receiver triangle, $\text{R}$, is formed by the three vertices, $\mathbf{q}_1$, $\mathbf{q}_2$, and $\mathbf{q}_3$.
The two triangles share an edge: $\mathbf{p}_2 = \mathbf{q}_3$ and $\mathbf{p}_3 = \mathbf{q}_2$.}
	\label{two_touch_problem}
\end{figure}

Consider the two-touch case in Fig.\ \ref{two_touch_problem}.
Without loss of generality, assume the midpoint of the edge that the two triangles share is located at the origin (i.e., $\mathbf{p}_2 + \mathbf{p}_3 = \mathbf{0}$).
The fact that they share an edge complicates the double surface integral over them.
Like in the one-touch case, we solve this problem by breaking the integral into several smaller integrals.

Let us look at $I^1$:
\begin{equation}
I^1 = I_{1^\prime\text{R}}^1 + I_{2^\prime2}^1 + I_{2^\prime3}^1 + I_{2^\prime4}^1 + I_{3^\prime2}^1 + I_{3^\prime3}^1 + I_{3^\prime4}^1 + I_{4^\prime2}^1 + I_{4^\prime3}^1 + I_{4^\prime4}^1 + I_{\text{S}1}^1 - I_{1^\prime1}^1.
\end{equation}
The integrals, $I_{1^\prime\text{R}}^1$, $I_{\text{S}1}^1$, and $I_{1^\prime1}^1$ are zero-touch integrals, and the integrals, $I_{2^\prime2}^1$, $I_{2^\prime4}^1$, $I_{3^\prime3}^1$, $I_{3^\prime4}^1$, $I_{4^\prime2}^1$, $I_{4^\prime3}^1$, and $I_{4^\prime4}^1$, are one-touch integrals.
Since we know how to compute them, let us combine them into a single integral:
\begin{equation}
\label{two_touch_1}
I^1 = I_{2^\prime3}^1 + I_{3^\prime2}^1 + I_{\text{remainder}}^1.
\end{equation}
That leaves $I_{2^\prime3}^1$ and $I_{3^\prime2}^1$, which have the same problem as the original integral: they correspond to pairs of triangles that share an edge.
Fortunately, the pair, $2^\prime$ and $3$, and the pair, $3^\prime$ and $2$, are each scaled and translated versions of the original pair, $\text{S}$ and $\text{R}$.
We derived earlier that
\begin{equation}
I_\ast^1 = s\left(\alpha\right)\alpha^4I^1.
\end{equation}
So,
\begin{equation}
I_{2^\prime3}^1 = s\left(\frac{1}{2}\right)\frac{1}{16}I^1,\quad
I_{3^\prime2}^1 = s\left(\frac{1}{2}\right)\frac{1}{16}I^1.
\end{equation}
Plugging these into Eq.\ (\ref{two_touch_1}) and rearranging,
\begin{equation}
I^1 = \left(1 - s\left(\frac{1}{2}\right)\frac{1}{8}\right)^{-1}I_{\text{remainder}}^1.
\end{equation}
The integrals, $I_{2^\prime3}^1$ and $I_{3^\prime2}^1$, are computed implicitly in this expression.

Now, let us look at $I^{\mathbf{p}^\prime}$:
\begin{equation}
I^{\mathbf{p}^\prime} = I_{1^\prime\text{R}}^{\mathbf{p}^\prime} + I_{2^\prime2}^{\mathbf{p}^\prime} + I_{2^\prime3}^{\mathbf{p}^\prime} + I_{2^\prime4}^{\mathbf{p}^\prime} + I_{3^\prime2}^{\mathbf{p}^\prime} + I_{3^\prime3}^{\mathbf{p}^\prime} + I_{3^\prime4}^{\mathbf{p}^\prime} + I_{4^\prime2}^{\mathbf{p}^\prime}+ I_{4^\prime3}^{\mathbf{p}^\prime} + I_{4^\prime4}^{\mathbf{p}^\prime} + I_{\text{S}1}^{\mathbf{p}^\prime} - I_{1^\prime1}^{\mathbf{p}^\prime}.
\end{equation}
Again, let us combine the ten integrals that we know how to compute into one:
\begin{equation}
\label{two_touch_pprime}
I^{\mathbf{p}^\prime} = I_{2^\prime3}^{\mathbf{p}^\prime} + I_{3^\prime2}^{\mathbf{p}^\prime} + I_{\text{remainder}}^{\mathbf{p}^\prime}.
\end{equation}
We derived earlier that
\begin{equation}
I_\ast^{\mathbf{p}^\prime} = s\left(\alpha\right)\alpha^5\left(I^{\mathbf{p}^\prime} + \left(\mathbf{p}^\prime\cdot\mathbf{t}\right)I^1\right).
\end{equation}
So,
\begin{equation}
I_{2^\prime3}^{\mathbf{p}^\prime} + I_{3^\prime2}^{\mathbf{p}^\prime} = s\left(\frac{1}{2}\right)\frac{1}{16}I^{\mathbf{p}^\prime} + s\left(\frac{1}{2}\right)\frac{1}{32}\left(\mathbf{p}^\prime\cdot\left(\mathbf{t}_{2^\prime3} + \mathbf{t}_{3^\prime2}\right)\right)I^1.
\end{equation}
We need to determine $\mathbf{t}_{2^\prime3}$ and $\mathbf{t}_{3^\prime2}$.
As discussed earlier, when $\alpha = 1 / 2$, the translation vector, $\mathbf{t}$, corresponds to the point that does not change during the transformation.
Thus, for the pair, $2^\prime$ and $3$, that is $\mathbf{p}_2$, and for the pair, $2$ and $3^\prime$, that is $\mathbf{p}_3$.
The edge that $\text{S}$ and $\text{R}$ share is centered around the origin, so $\mathbf{t}_{2^\prime3} + \mathbf{t}_{3^\prime2} = \mathbf{0}$, which means that
\begin{equation}
I_{2^\prime3}^{\mathbf{p}^\prime} + I_{3^\prime2}^{\mathbf{p}^\prime} = s\left(\frac{1}{2}\right)\frac{1}{16}I^{\mathbf{p}^\prime}.
\end{equation}
Plugging this into Eq.\ (\ref{two_touch_pprime}) and rearranging,
\begin{equation}
I^{\mathbf{p}^\prime} = \left(1 - s\left(\frac{1}{2}\right)\frac{1}{16}\right)^{-1}I_{\text{remainder}}^{\mathbf{p}^\prime}.
\end{equation}

The same analysis can be used for comuting $I^{\mathbf{p}}$:
\begin{equation}
I^{\mathbf{p}} = \left(1 - s\left(\frac{1}{2}\right)\frac{1}{16}\right)^{-1}I_{\text{remainder}}^{\mathbf{p}}.
\end{equation}

Finally, let us look at $I^{\mathbf{p}^\prime\mathbf{p}}$:
\begin{equation}
\label{two_touch_pprimep}
I^{\mathbf{p}^\prime\mathbf{p}} = I_{2^\prime3}^{\mathbf{p}^\prime\mathbf{p}} + I_{3^\prime2}^{\mathbf{p}^\prime\mathbf{p}} + I_{\text{remainder}}^{\mathbf{p}^\prime\mathbf{p}},
\end{equation}
where
\begin{equation}
I_{\text{remainder}}^{\mathbf{p}^\prime\mathbf{p}} = I_{1^\prime\text{R}}^{\mathbf{p}^\prime\mathbf{p}} + I_{2^\prime2}^{\mathbf{p}^\prime\mathbf{p}} + I_{2^\prime4}^{\mathbf{p}^\prime\mathbf{p}} + I_{3^\prime3}^{\mathbf{p}^\prime\mathbf{p}} + I_{3^\prime4}^{\mathbf{p}^\prime\mathbf{p}} + I_{4^\prime2}^{\mathbf{p}^\prime\mathbf{p}}+ I_{4^\prime3}^{\mathbf{p}^\prime\mathbf{p}} + I_{4^\prime4}^{\mathbf{p}^\prime\mathbf{p}} + I_{\text{S}1}^{\mathbf{p}^\prime\mathbf{p}} - I_{1^\prime1}^{\mathbf{p}^\prime\mathbf{p}}.
\end{equation}
We derived earlier that
\begin{equation}
I_\ast^{\mathbf{p}^\prime\mathbf{p}} = s\left(\alpha\right)\alpha^6\left(
I^{\mathbf{p}^\prime\mathbf{p}} +
\left(\mathbf{p}\cdot\mathbf{t}\right)I^{\mathbf{p}^\prime} +
\left(\mathbf{p}^\prime\cdot\mathbf{t}\right)I^{\mathbf{p}} +
\left(\mathbf{p}\cdot\mathbf{t}\right)\left(\mathbf{p}^\prime\cdot\mathbf{t}\right)I^1
\right).
\end{equation}
So,
\begin{equation}
\label{two_touch_pprimep_relation}
I_{2^\prime3}^{\mathbf{p}^\prime\mathbf{p}} + I_{3^\prime2}^{\mathbf{p}^\prime\mathbf{p}} =
s\left(\frac{1}{2}\right)\frac{1}{32}I^{\mathbf{p}^\prime\mathbf{p}} +
s\left(\frac{1}{2}\right)\frac{1}{64}\left(\left(\mathbf{p}\cdot\mathbf{t}_{2^\prime3}\right)\left(\mathbf{p}^\prime\cdot\mathbf{t}_{2^\prime3}\right) + \left(\mathbf{p}\cdot\mathbf{t}_{3^\prime2}\right)\left(\mathbf{p}^\prime\cdot\mathbf{t}_{3^\prime2}\right)\right)I^1.
\end{equation}
The terms linear in $\mathbf{p}^\prime$ and $\mathbf{p}$ disappear because $\mathbf{t}_{2^\prime3} + \mathbf{t}_{3^\prime2} = \mathbf{0}$, but unfortunately, the other terms remain.
Plugging Eq.\ (\ref{two_touch_pprimep_relation}) into Eq.\ (\ref{two_touch_pprimep}) and rearranging,
\begin{equation}
I^{\mathbf{p}^\prime\mathbf{p}} = \left(1 - s\left(\frac{1}{2}\right)\frac{1}{32}\right)^{-1}\left(I_{\text{remainder}}^{\mathbf{p}^\prime\mathbf{p}} + a^{\mathbf{p}^\prime\mathbf{p}}\right),
\end{equation}
where
\begin{equation}
a^{\mathbf{p}^\prime\mathbf{p}} = s\left(\frac{1}{2}\right)\frac{1}{64}\left(\left(\mathbf{p}\cdot\mathbf{t}_{2^\prime3}\right)\left(\mathbf{p}^\prime\cdot\mathbf{t}_{2^\prime3}\right) + \left(\mathbf{p}\cdot\mathbf{t}_{3^\prime2}\right)\left(\mathbf{p}^\prime\cdot\mathbf{t}_{3^\prime2}\right)\right)I^1.
\end{equation}

\subsection{Three-Touch Case}
\label{threetouch}

\begin{figure}[t]
	\centering
	\includegraphics[scale=0.75]{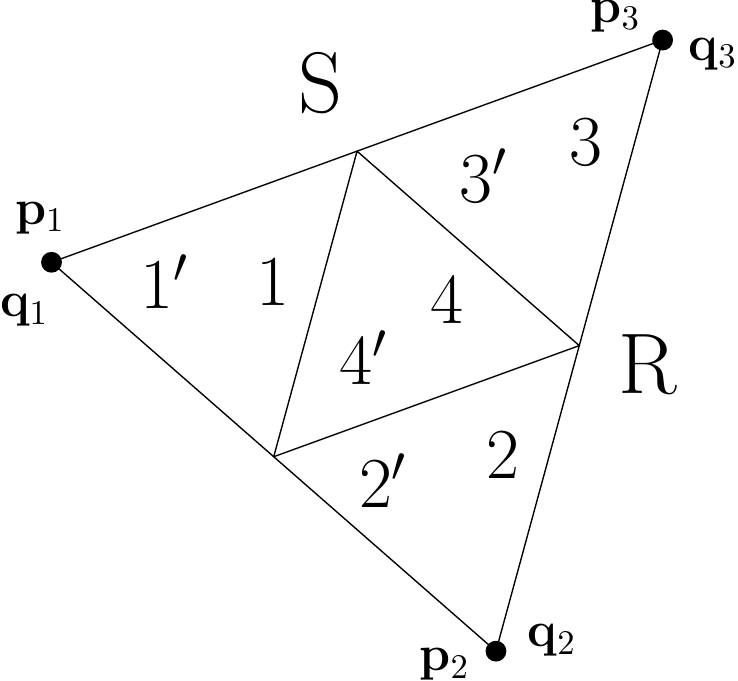}
	\caption{A diagram showing how to compute the double surface integral for two triangles that are the same.
The source triangle, $\text{S}$, is formed by the three vertices, $\mathbf{p}_1$, $\mathbf{p}_2$, $\mathbf{p}_3$, and the receiver triangle, $\text{R}$, is formed by the three vertices, $\mathbf{q}_1$, $\mathbf{q}_2$, and $\mathbf{q}_3$.
The two triangles are the same: $\mathbf{p}_1 = \mathbf{q}_1$, $\mathbf{p}_2 = \mathbf{q}_2$, and $\mathbf{p}_3 = \mathbf{q}_3$.}
	\label{three_touch_problem}
\end{figure}

Consider the three-touch case in Fig.\ \ref{three_touch_problem}.
In this case, the two triangles are the same.
Without loss of generality, assume that the centroids of the two triangles are located at the origin (i.e., $\mathbf{p}_1 + \mathbf{p}_2 + \mathbf{p}_3 = \mathbf{0}$).
The fact that the two triangles are the same complicates the double surface integral over them.
Like in the one- and two-touch cases, we solve this problem by breaking the integral into several smaller integrals.

Let us look at $I^1$:
\begin{equation}
I^1 =
I_{1^\prime1}^1 + I_{1^\prime2}^1 + I_{1^\prime3}^1 + I_{1^\prime4}^1 +
I_{2^\prime1}^1 + I_{2^\prime2}^1 + I_{2^\prime3}^1 + I_{2^\prime4}^1 +
I_{3^\prime1}^1 + I_{3^\prime2}^1 + I_{3^\prime3}^1 + I_{3^\prime4}^1 +
I_{4^\prime1}^1 + I_{4^\prime2}^1 + I_{4^\prime3}^1 + I_{4^\prime4}^1.
\end{equation}
The integrals, $I_{1^\prime2}^1$, $I_{1^\prime3}^1$, $I_{2^\prime1}^1$, $I_{2^\prime3}^1$, $I_{3^\prime1}^1$, and $I_{3^\prime2}^1$, are one-touch integrals, and the integrals, $I_{1^\prime4}^1$, $I_{2^\prime4}^1$, $I_{3^\prime4}^1$, $I_{4^\prime1}^1$, $I_{4^\prime2}^1$, and $I_{4^\prime3}^1$, are two-touch integrals.
Like before, because we know how to compute them, let us combine these 12 integrals into a single integral:
\begin{equation}
\label{three_touch_1}
I^1 =
I_{1^\prime1}^1 + I_{2^\prime2}^1 + I_{3^\prime3}^1 + I_{4^\prime4}^1 + I_{\text{remainder}}^1.
\end{equation}
That leaves the four integrals, $I_{1^\prime1}^1$, $I_{2^\prime2}^1$, $I_{3^\prime3}^1$, and $I_{4^\prime4}^1$, which have the same problem as the original: they correspond to pairs of triangles that are the same.
We derived earlier that
\begin{equation}
I_\ast^1 = s\left(\alpha\right)\alpha^4I^1.
\end{equation}
So,
\begin{equation}
I_{1^\prime1}^1 + I_{2^\prime2}^1 + I_{3^\prime3}^1 = s\left(\frac{1}{2}\right)\frac{3}{16}I^1.
\end{equation}
The pair of triangles, $4^\prime$ and $4$, are not only scaled by a factor of $1 / 2$, but also rotated by 180 degrees.
Because they are centered around the origin, we can achieve this rotation by setting $\alpha = -1 / 2$:
\begin{equation}
I_{4^\prime4}^1 = s\left(-\frac{1}{2}\right)\frac{1}{16}I^1.
\end{equation}
Plugging these into Eq.\ (\ref{three_touch_1}) and rearranging,
\begin{equation}
I^1 = \left(1 - s\left(\frac{1}{2}\right)\frac{3}{16} - s\left(-\frac{1}{2}\right)\frac{1}{16}\right)^{-1}I_{\text{remainder}}^1.
\end{equation}

The same analysis from here and before can be used to compute expressions for $I^{\mathbf{p}^\prime}$, $I^{\mathbf{p}}$, and $I^{\mathbf{p}^\prime\mathbf{p}}$:
\begin{equation}
I^{\mathbf{p}^\prime} = \left(1 - s\left(\frac{1}{2}\right)\frac{3}{32} + s\left(-\frac{1}{2}\right)\frac{1}{32}\right)^{-1}I_{\text{remainder}}^{\mathbf{p}^\prime},
\end{equation}
\begin{equation}
I^{\mathbf{p}} = \left(1 - s\left(\frac{1}{2}\right)\frac{3}{32} + s\left(-\frac{1}{2}\right)\frac{1}{32}\right)^{-1}I_{\text{remainder}}^{\mathbf{p}},
\end{equation}
\begin{equation}
I^{\mathbf{p}^\prime\mathbf{p}} = \left(1 - s\left(\frac{1}{2}\right)\frac{3}{64} - s\left(-\frac{1}{2}\right)\frac{1}{64}\right)^{-1}\left(I_{\text{remainder}}^{\mathbf{p}^\prime\mathbf{p}} + a^{\mathbf{p}^\prime\mathbf{p}}\right),
\end{equation}
where
\begin{equation}
a^{\mathbf{p}^\prime\mathbf{p}} = s\left(\frac{1}{2}\right)\frac{1}{64}\left(
\left(\mathbf{p}\cdot\mathbf{t}_{1^\prime1}\right)\left(\mathbf{p}^\prime\cdot\mathbf{t}_{1^\prime1}\right) +
\left(\mathbf{p}\cdot\mathbf{t}_{2^\prime2}\right)\left(\mathbf{p}^\prime\cdot\mathbf{t}_{2^\prime2}\right) +
\left(\mathbf{p}\cdot\mathbf{t}_{3^\prime3}\right)\left(\mathbf{p}^\prime\cdot\mathbf{t}_{3^\prime3}\right)
\right)I^1,
\end{equation}
and $\mathbf{t}_{1^\prime1} = \mathbf{p}_1$, $\mathbf{t}_{2^\prime2} = \mathbf{p}_2$, and $\mathbf{t}_{3^\prime3} = \mathbf{p}_3$.

\section{Numerical Examples}

We have implemented the two methods that were described in Secs.\ \ref{zerocase} and \ref{onetwothreecase} in MATLAB as part of a Galerkin BEM library for the Laplace equation.
Currently, the library only supports external problems with Dirichlet boundary conditions, but even such simple problems provide an opportunity for demonstrating these two methods.
The library supports constant and linear triangular elements, and for comparison purposes, we have also implemented the collocation method.
Thus, the library contains four solvers: (1) constant collocation; (2) linear collocation; (3) constant Galerkin; and (4) linear Galerkin.
In this section, we demonstrate that the two methods work as advertised by using the library to solve two test problems.
To begin, we mention a couple implementation details of the library.
Following that, we describe and give the analytical solutions to the test problems, and then show some results from solving the problems using the library.

\subsection{Implementation Details}

There are two implementation details that should be mentioned.

First, in MATLAB, for the same accuracy, the analytical method for the zero-touch case described in Sec.\ \ref{zerocase} is faster for pairs of triangles that are far away from each other, and the semi-analytical method is faster for pairs of triangles that are nearby each other.
Thus, when a zero-touch integral needs to be computed directly, and not indirectly as part of a one-, two-, or three-touch integral, the analytical method is used.
However, the semi-analytical method is used instead when a zero-touch integral needs to be computed during the computation of a one-, two-, or three-touch integral.
For example, when computing a one-touch integral, six zero-touch integrals need to be computed, and these are done so semi-analytically.

Second, for external problems with Dirichlet boundary conditions, the only kernel that needs to be integrated is the Laplace equation's Green's function.
That is, when $F\left(\mathbf{r}\right) = G\left(\mathbf{r}\right)$, where
\begin{equation}
G\left(\mathbf{r}\right) = \frac{1}{4\pi\left|\mathbf{r}\right|}.
\end{equation}
The associated scaling function, $s\left(\alpha\right)$, for this kernel is given by
\begin{equation}
s\left(\alpha\right) = \frac{1}{\left|\alpha\right|}.
\end{equation}
Below, we quickly specialize and summarize the expressions derived in Sec.\ \ref{onetwothreecase} for this kernel.

For the one-touch case:
\begin{equation}
I^1 = \frac{8}{7}I_{\text{remainder}}^1
,\quad
I^{\mathbf{p}^\prime} = \frac{16}{15}I_{\text{remainder}}^{\mathbf{p}^\prime}
,\quad
I^{\mathbf{p}} = \frac{16}{15}I_{\text{remainder}}^{\mathbf{p}}
,\quad
I^{\mathbf{p}^\prime\mathbf{p}} = \frac{32}{31}I_{\text{remainder}}^{\mathbf{p}^\prime\mathbf{p}}.
\end{equation}

For the two-touch case:
\begin{equation}
I^1 = \frac{4}{3}I_{\text{remainder}}^1
,\quad
I^{\mathbf{p}^\prime} = \frac{8}{7}I_{\text{remainder}}^{\mathbf{p}^\prime}
,\quad
I^{\mathbf{p}} = \frac{8}{7}I_{\text{remainder}}^{\mathbf{p}}
,\quad
I^{\mathbf{p}^\prime\mathbf{p}} = \frac{16}{15}\left(I_{\text{remainder}}^{\mathbf{p}^\prime\mathbf{p}} + a^{\mathbf{p}^\prime\mathbf{p}}\right),
\end{equation}
where
\begin{equation}
a^{\mathbf{p}^\prime\mathbf{p}} = \frac{1}{32}\left(\left(\mathbf{p}\cdot\mathbf{t}_{2^\prime3}\right)\left(\mathbf{p}^\prime\cdot\mathbf{t}_{2^\prime3}\right) + \left(\mathbf{p}\cdot\mathbf{t}_{3^\prime2}\right)\left(\mathbf{p}^\prime\cdot\mathbf{t}_{3^\prime2}\right)\right)I^1.
\end{equation}

For the three-touch case:
\begin{equation}
I^1 = 2I_{\text{remainder}}^1
,\quad
I^{\mathbf{p}^\prime} = \frac{8}{7}I_{\text{remainder}}^{\mathbf{p}^\prime}
,\quad
I^{\mathbf{p}} = \frac{8}{7}I_{\text{remainder}}^{\mathbf{p}}
,\quad
I^{\mathbf{p}^\prime\mathbf{p}} = \frac{8}{7}\left(I_{\text{remainder}}^{\mathbf{p}^\prime\mathbf{p}} + a^{\mathbf{p}^\prime\mathbf{p}}\right),
\end{equation}
where
\begin{equation}
a^{\mathbf{p}^\prime\mathbf{p}} = \frac{1}{32}\left(
\left(\mathbf{p}\cdot\mathbf{t}_{1^\prime1}\right)\left(\mathbf{p}^\prime\cdot\mathbf{t}_{1^\prime1}\right) +
\left(\mathbf{p}\cdot\mathbf{t}_{2^\prime2}\right)\left(\mathbf{p}^\prime\cdot\mathbf{t}_{2^\prime2}\right) +
\left(\mathbf{p}\cdot\mathbf{t}_{3^\prime3}\right)\left(\mathbf{p}^\prime\cdot\mathbf{t}_{3^\prime3}\right)
\right)I^1.
\end{equation}

\subsection{Test Problem \#1: Cube Behaving Like a Point Source}

\begin{figure}[t]
	\centering
	\includegraphics[scale=0.4]{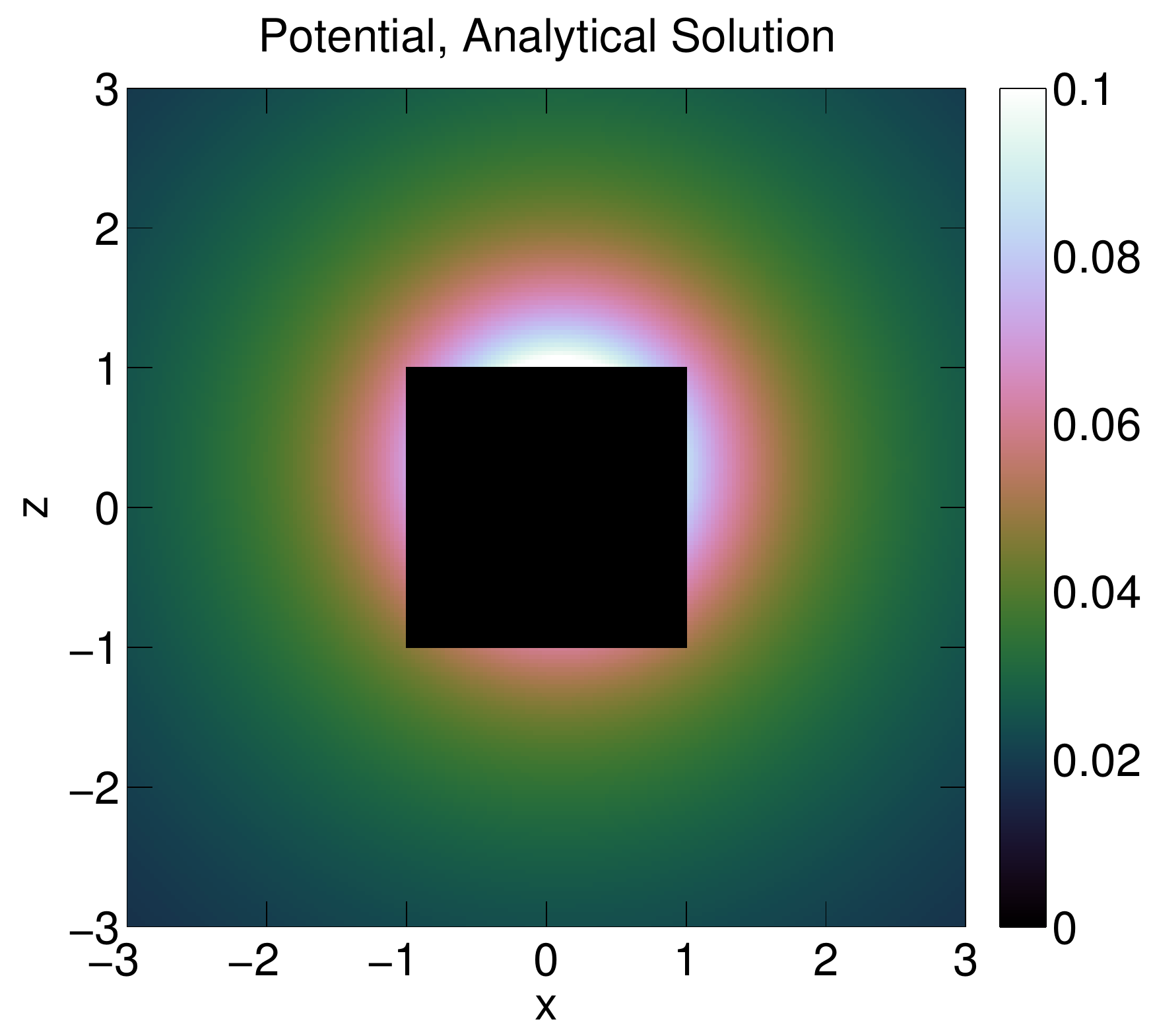}
	\includegraphics[scale=0.4]{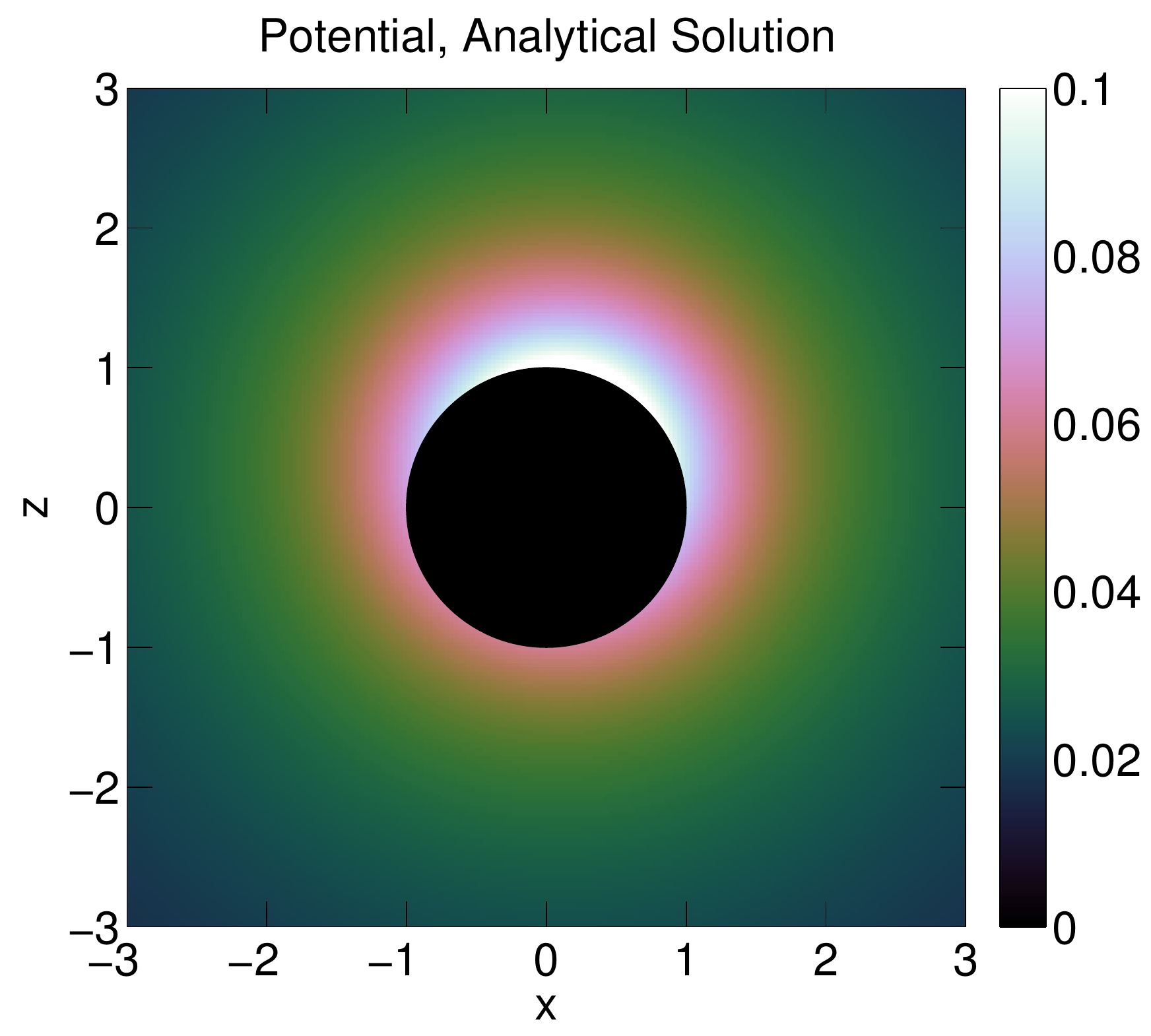}
	\caption{The potential due to a point source placed inside a cube (left) or a sphere (right).}
	\label{test1_ana}
\end{figure}

Consider a cube, $\text{C}$, centered around the origin, and place a point source at $\mathbf{x}_0$, where $\mathbf{x}_0$ is inside the cube.
The potential due to the point source is given by
\begin{equation}
\label{ps_potential}
\phi^\text{ps}\left(\mathbf{x}\right) = \frac{1}{4\pi\left|\mathbf{x} - \mathbf{x}_0\right|}
\end{equation}
and is shown in Fig.\ \ref{test1_ana}.
Suppose we want to compute an equivalent source distribution on the surface of the cube that gives rise to the same potential outside the cube as that of the pont source.
In other words, we want to compute a surface source distribution on the cube such that
\begin{equation}
\phi\left(\mathbf{x}\right) = \phi^\text{ps}\left(\mathbf{x}\right) = \frac{1}{4\pi\left|\mathbf{x} - \mathbf{x}_0\right|},\quad\mathbf{x} \in \partial\text{C},
\end{equation}
where $\partial\text{C}$ is the surface of the cube.
Obviously, the solution is $\phi\left(\mathbf{x}\right) = \phi^\text{ps}\left(\mathbf{x}\right)$.

\subsection{Test Problem \#2: Sphere Behaving Like a Point Source}

The sphere test problem is the same as the cube test problem, except that the boundary is a sphere instead of a cube.
In other words, we want to compute a source distribution on the surface of the sphere that gives rise to a potential outside of the sphere that is the same as that due to a point source placed inside the sphere (see Fig.\ \ref{test1_ana}).
The point source's position is the same in both test problems.
The only difference is the geometry of the boundary.

\subsection{Error Analysis}

\begin{figure}[t]
	\centering
	\includegraphics[scale=0.4]{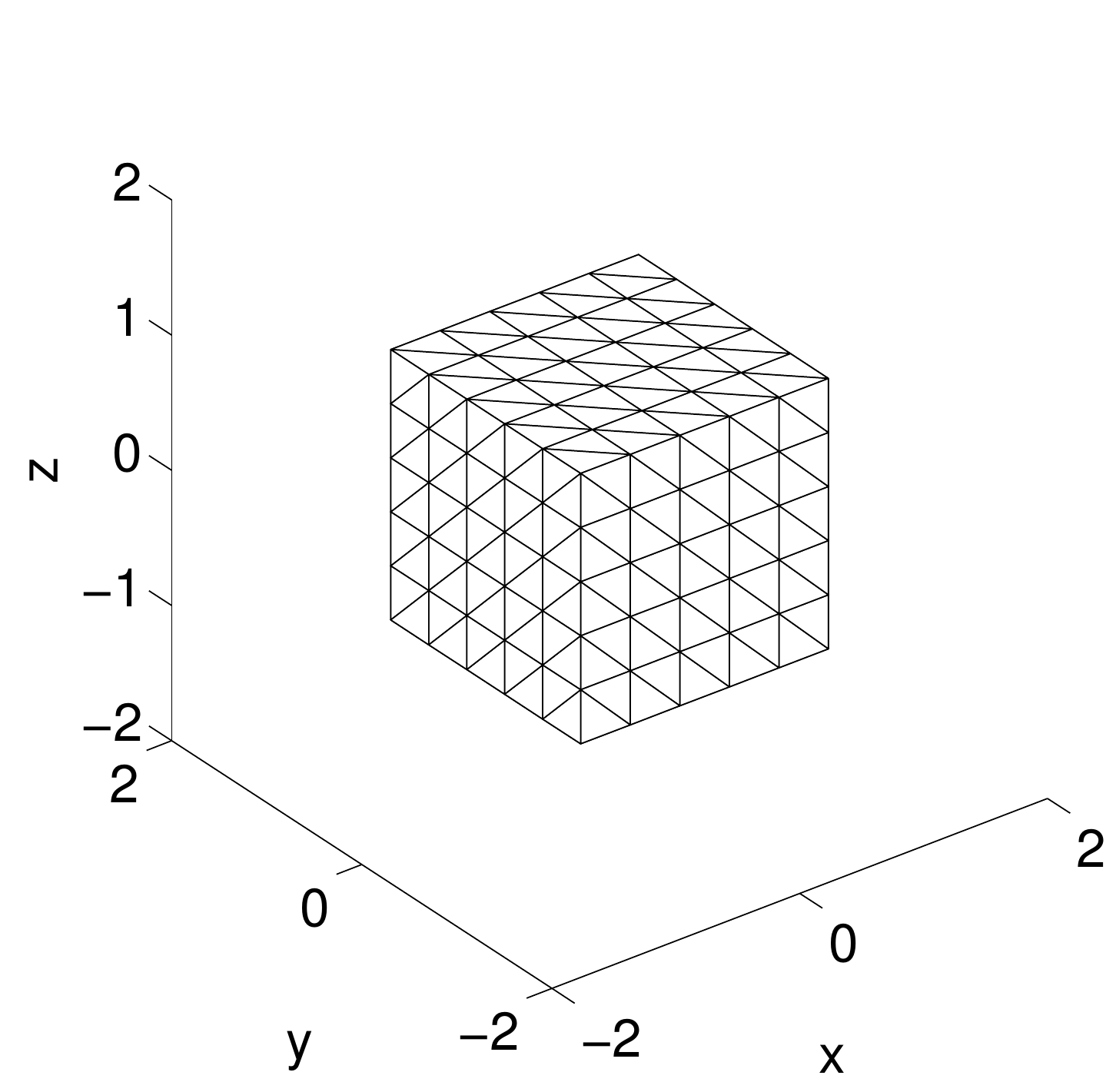}
	\includegraphics[scale=0.4]{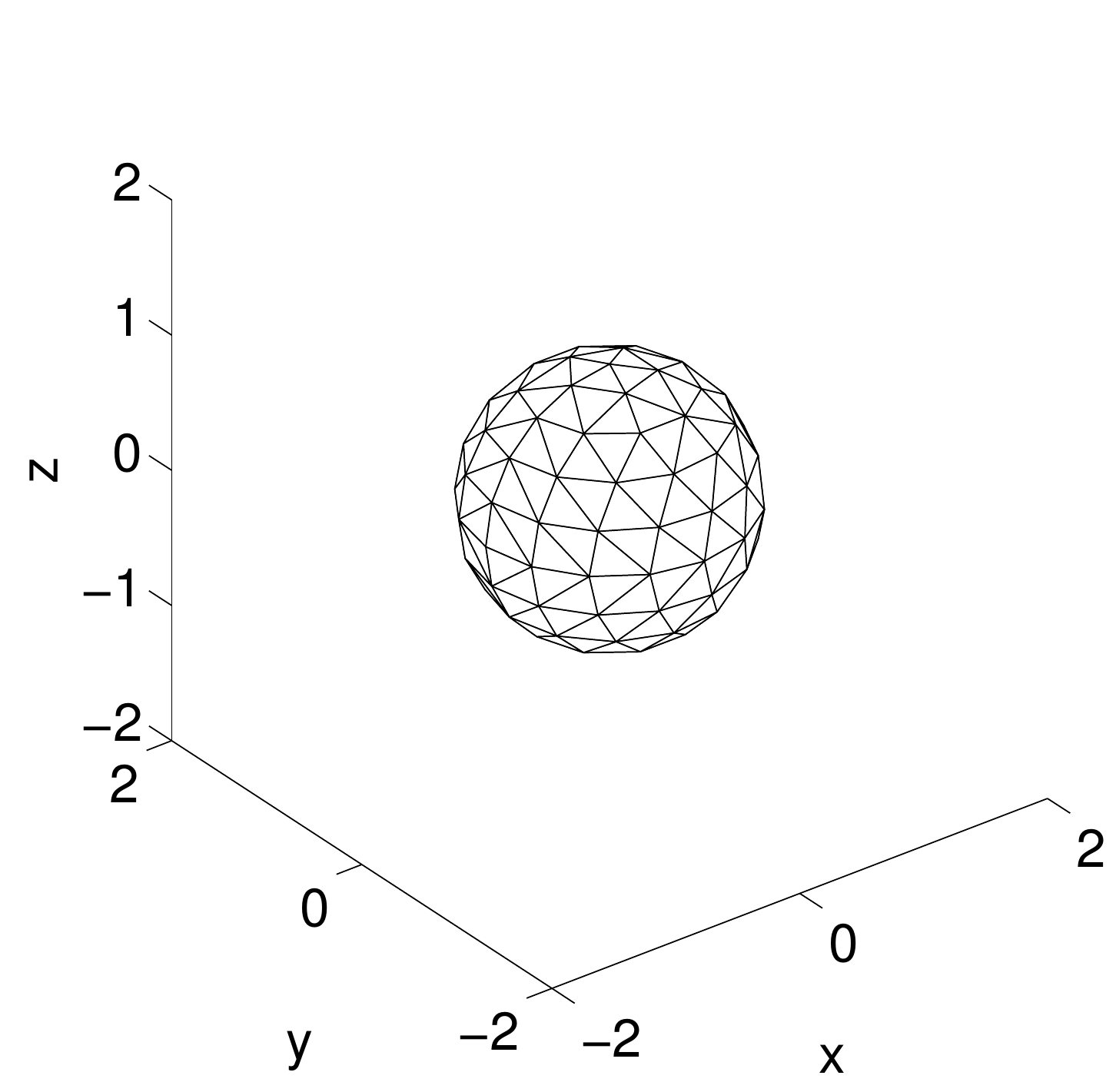}
	\caption{Two meshes used to compute the numerical solutions to the two test problems.  On the left: a mesh for the cube (300 triangles).  On the right: a mesh for the sphere (204 triangles).}
	\label{test_meshes}
\end{figure}

\begin{figure}[t]
	\centering
	\includegraphics[scale=0.4]{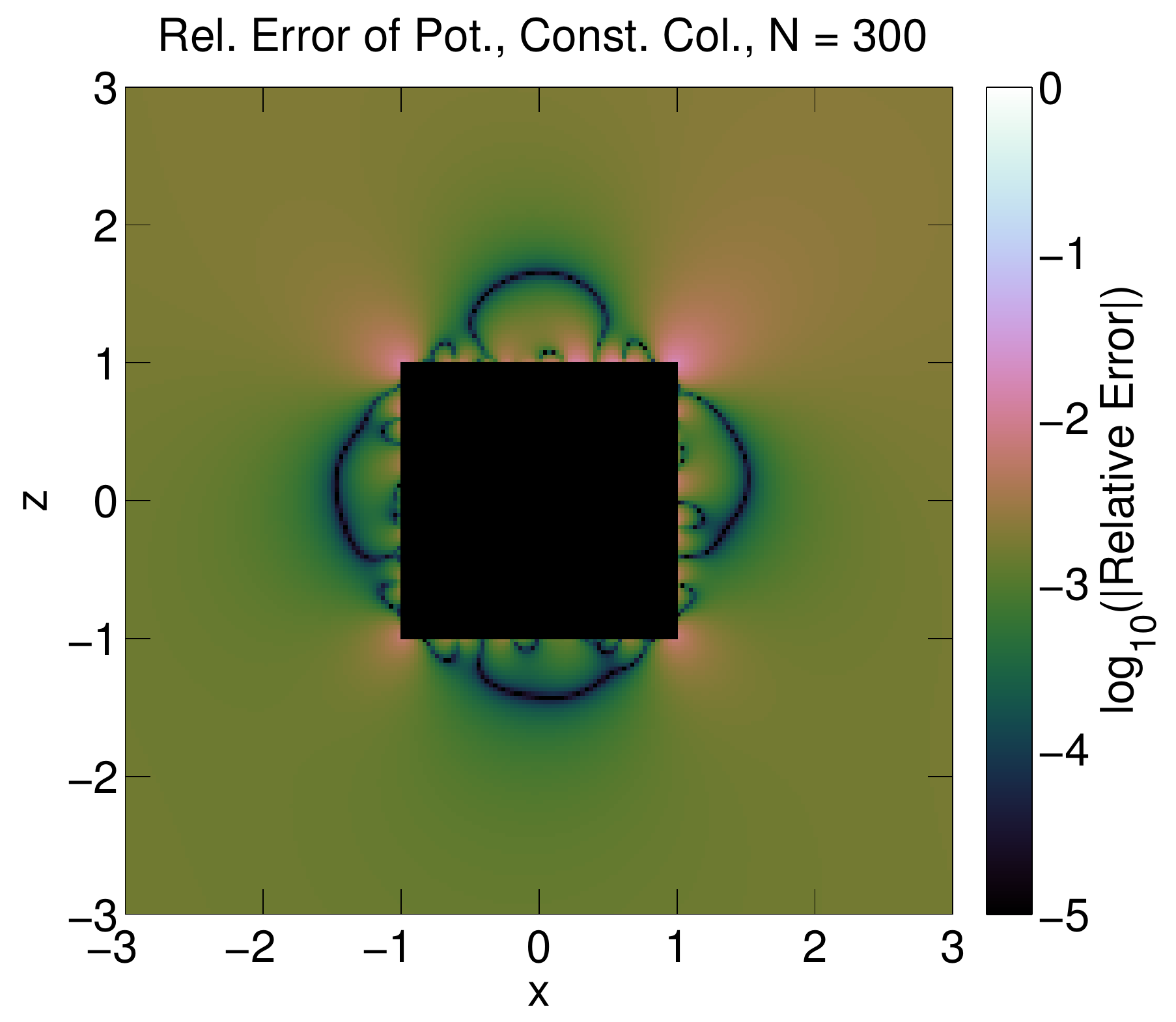}
	\includegraphics[scale=0.4]{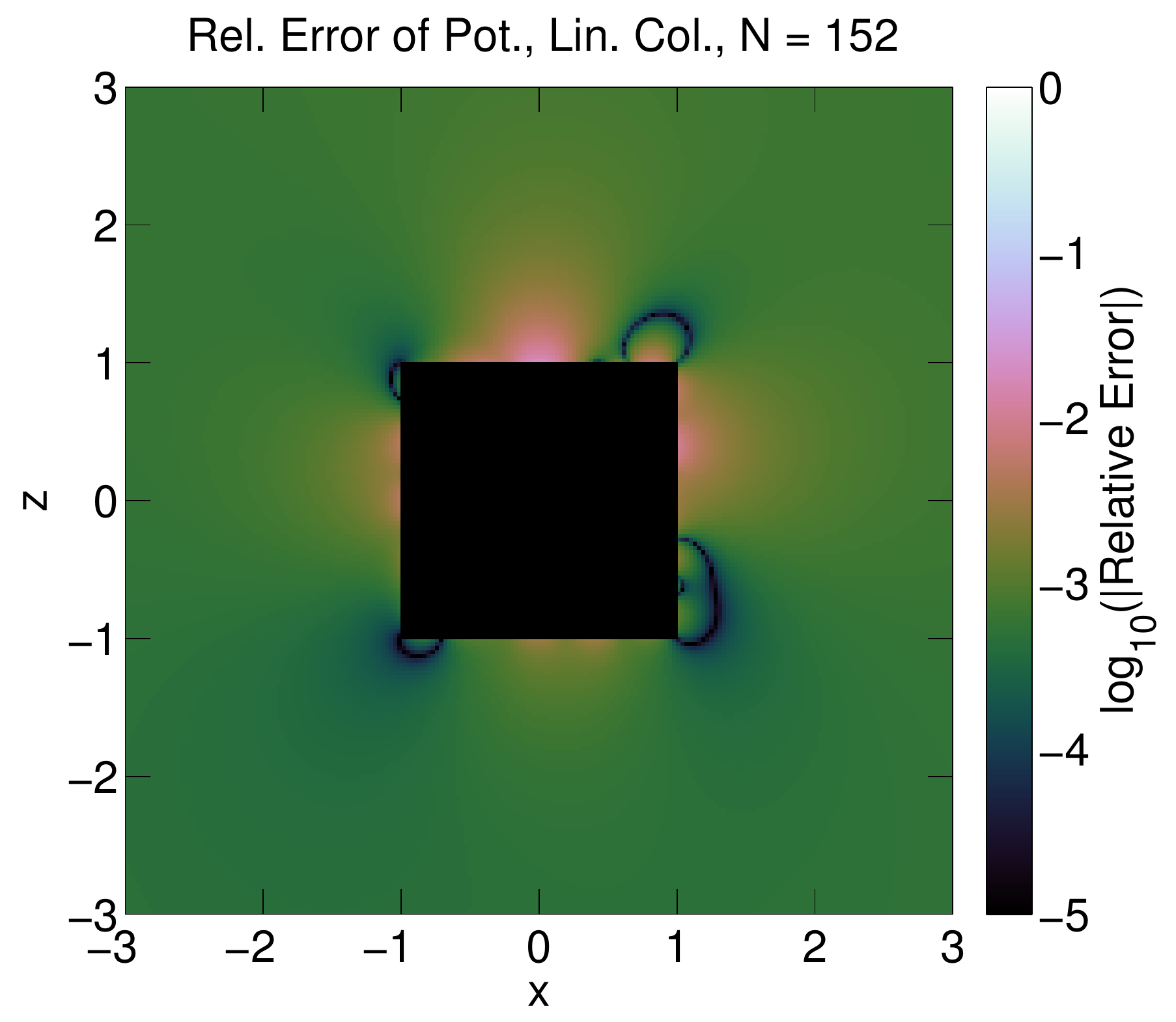}\\
	\includegraphics[scale=0.4]{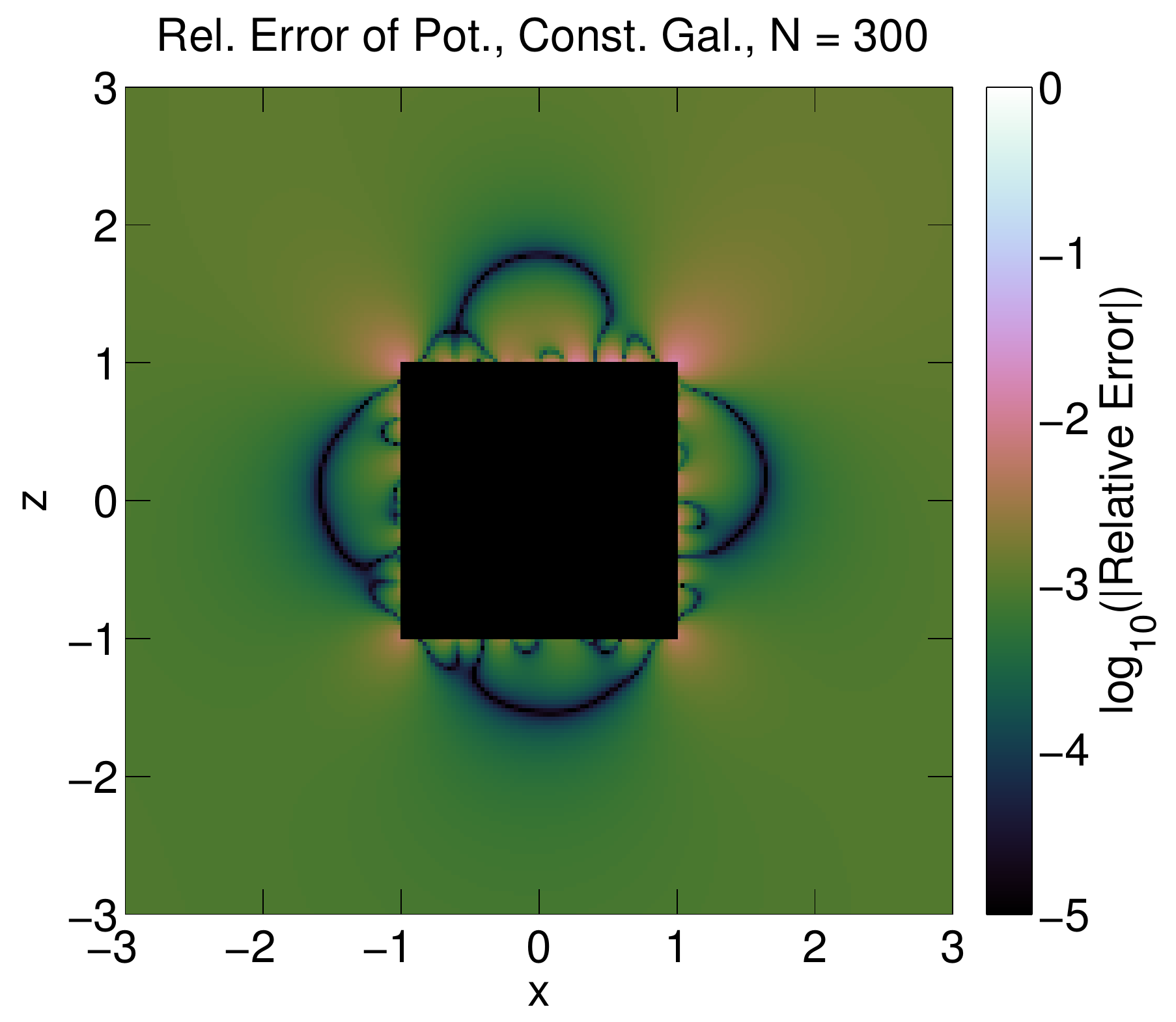}
	\includegraphics[scale=0.4]{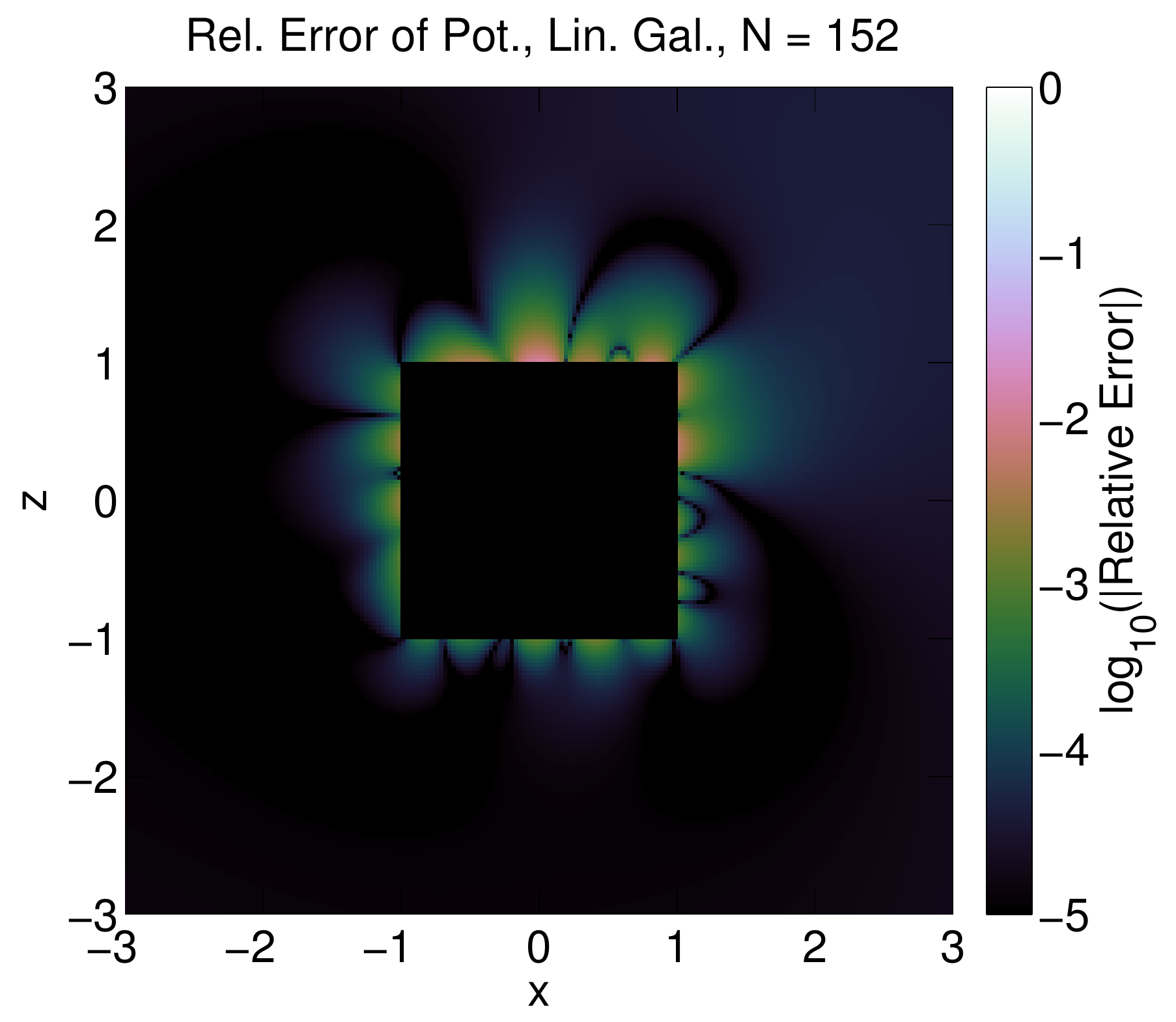}
	\caption{The relative error of the numerical solutions returned by the four solvers for the cube test problem using the mesh seen in Fig.\ \ref{test_meshes}.  The value for $N$ shown in the titles is the number of unknowns.  For constant elements, that is the number of triangles.  For linear elements, that is the number of vertices.}
	\label{test1_rel_error}
\end{figure}

\begin{figure}[t]
	\centering
	\includegraphics[scale=0.4]{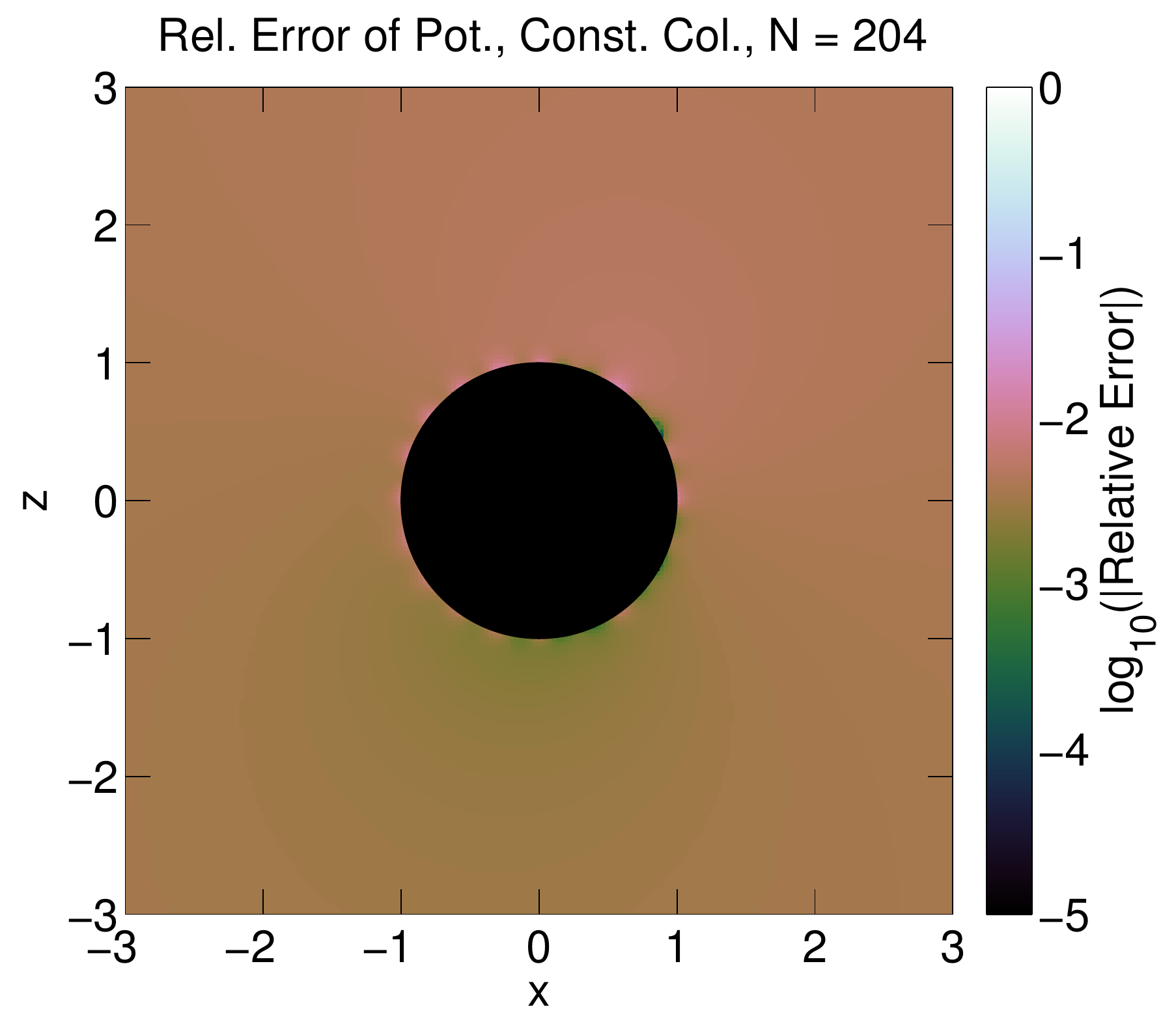}
	\includegraphics[scale=0.4]{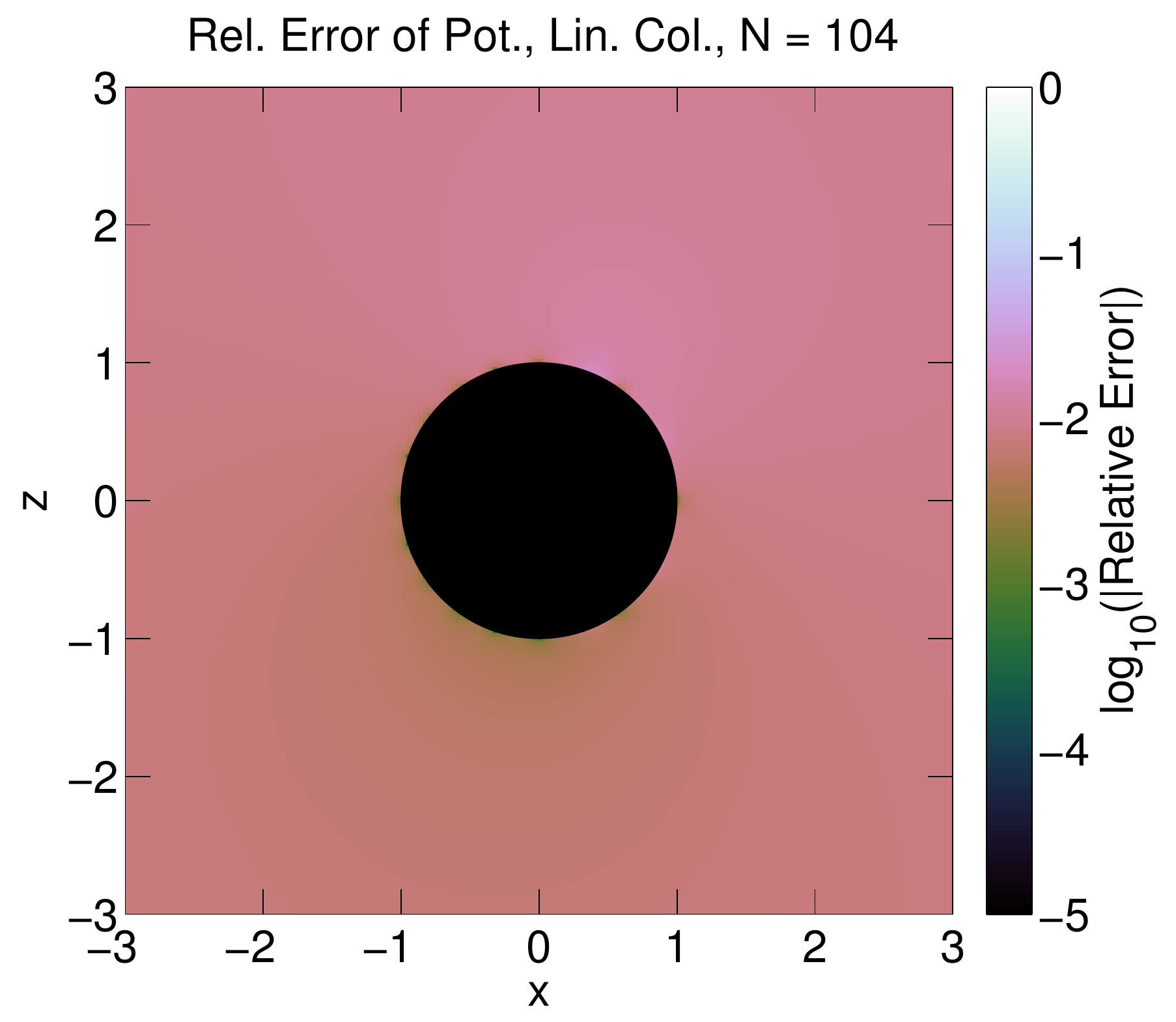}\\
	\includegraphics[scale=0.4]{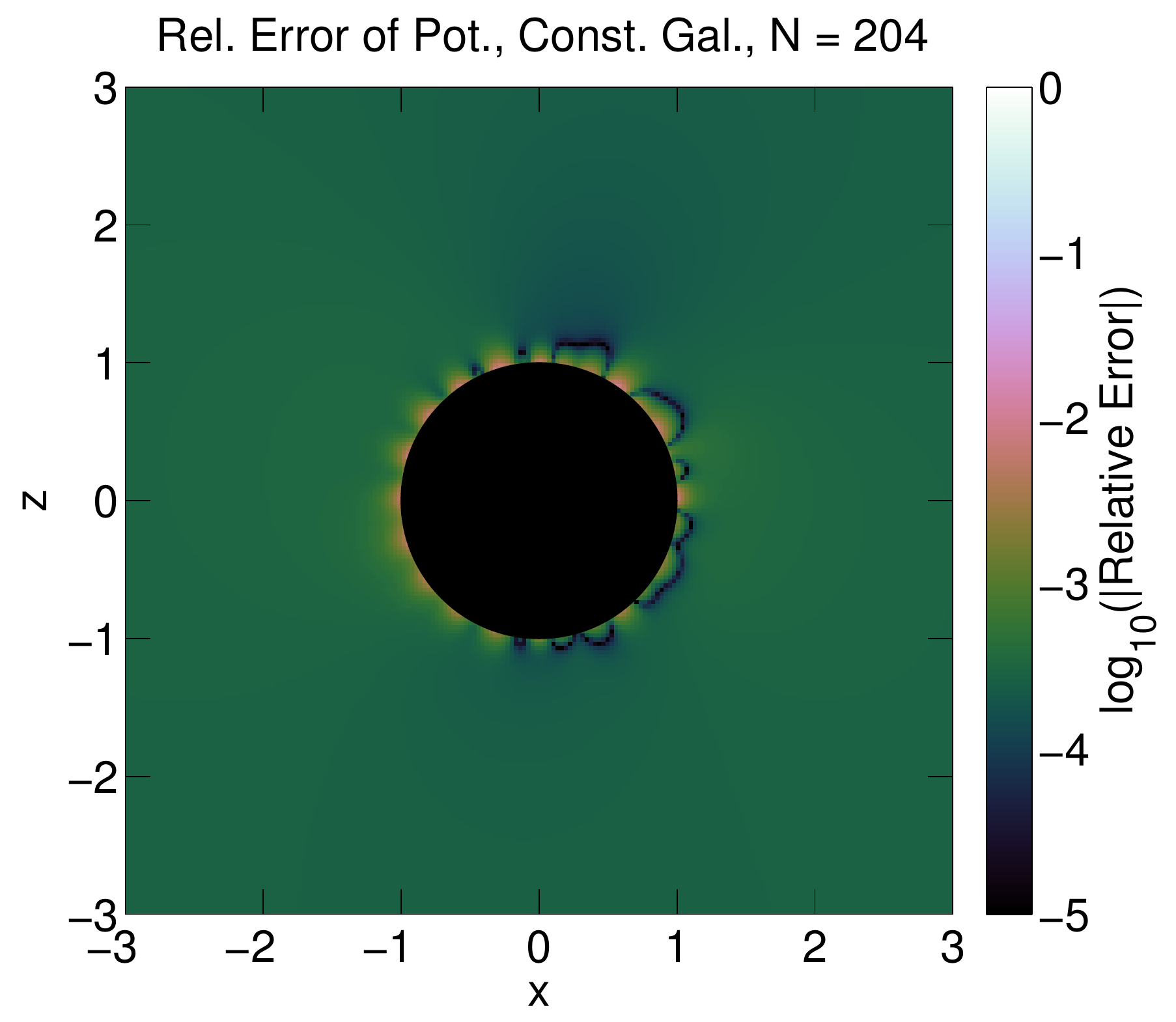}
	\includegraphics[scale=0.4]{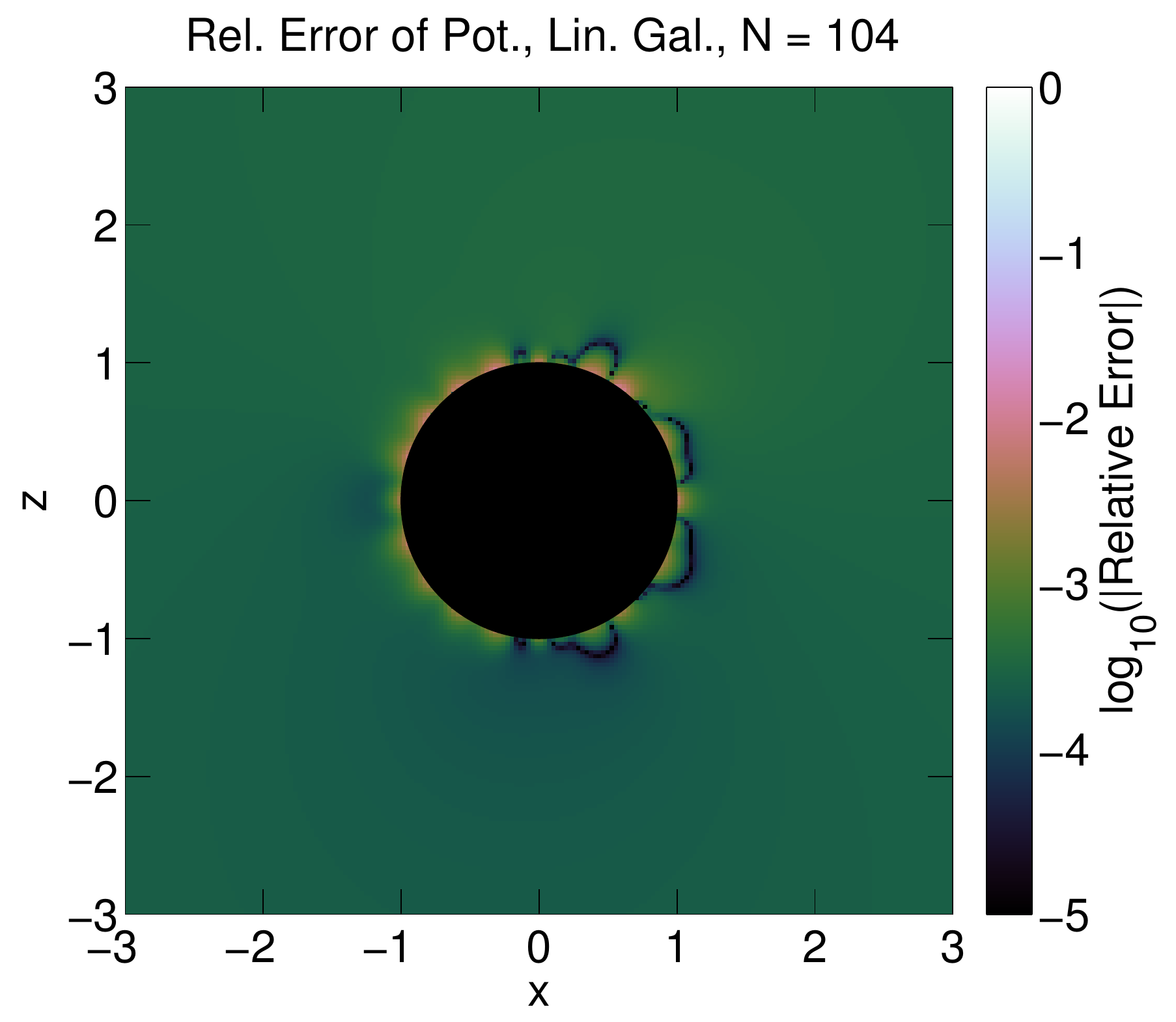}
	\caption{The relative error of the numerical solutions returned by the four solvers for the sphere test problem using the mesh seen in Fig.\ \ref{test_meshes}.}
	\label{test2_rel_error}
\end{figure}

\begin{figure}[t]
	\centering
	\includegraphics[scale=0.4]{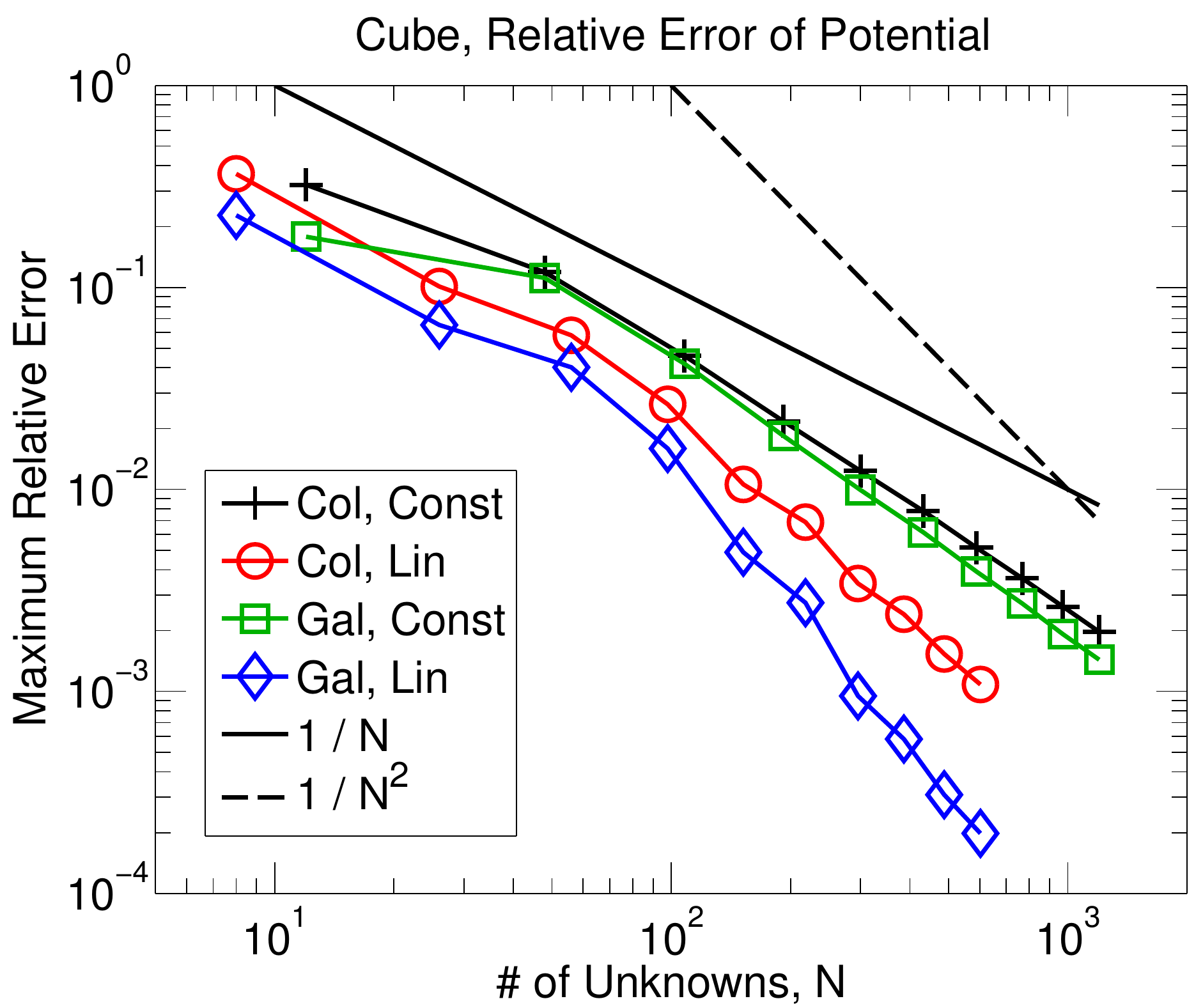}
	\includegraphics[scale=0.4]{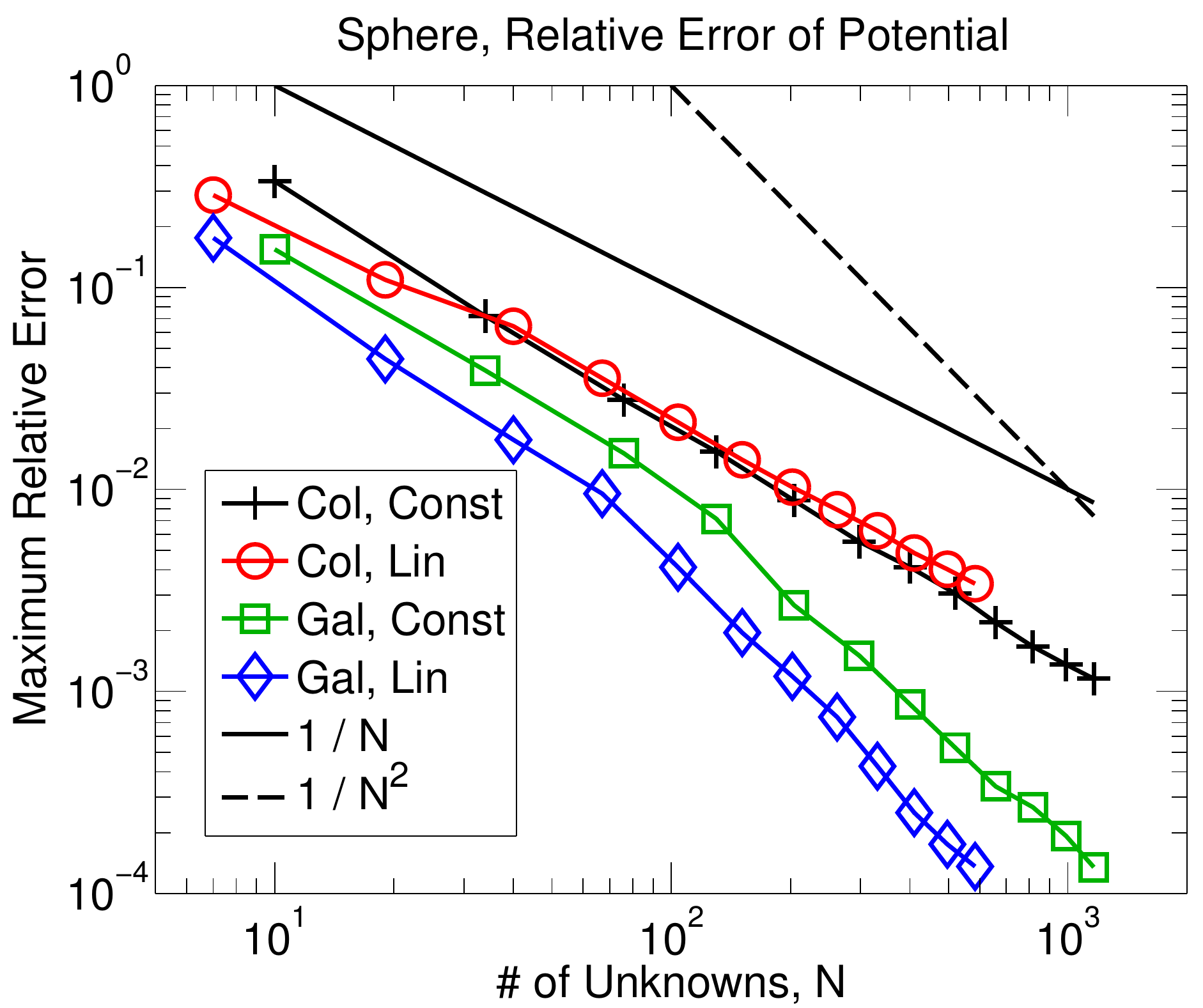}
	\caption{The maximum relative error as a function of the number of unknowns of the four solvers for the cube test problem (left) and the sphere test problem (right).}
	\label{test1_max_rel_error}
\end{figure}

First, we had to construct triangular meshes for the two test problems.
Fig.\ \ref{test_meshes} shows two example meshes, one for each test problem.
The meshes were constructed so that all the triangles were as close to equilateral and the same size as each other as possible.
In fact, in the cube mesh, all the triangles are congruent.

Second, we solved the two test problems using the four solvers.
The software performed very well for the cube test problem.
In particular, the linear Galerkin solver performed extremely well.
Fig.\ \ref{test1_rel_error} shows the relative error in the volume around the cube for each solver when the mesh had 300 triangles.
To show that the relative errors decrease as the mesh size increases, we computed and plotted the maximum relative error in the volume around the cube as a function of the number of unknowns, $N$ (see Fig.\ \ref{test1_max_rel_error}).
For constant triangular elements, the number of unknowns is the number of triangles.
For linear triangular elements, the number of unknowns is the number of vertices.
All four curves behave at least as good as $1 / N$.
In general, linear triangular elements are better than constant triangular elements, and the Galerkin methods are better than the collocation methods.
The linear Galerkin method behaves even better than $1 / N$ for larger $N$.
The nice thing about the cube test problem is that the boundary is piecewise planar, and so can be modeled exactly using planar triangular elements.
Thus, there are no errors from approximating the geometry of the boundary.

The software performed very well for the sphere test problem as well.
Fig.\ \ref{test2_rel_error} shows the relative error in the volume around the sphere for each solver when the mesh had 204 triangles.
Fig.\ \ref{test1_max_rel_error} shows the maximum relative error as a function of the number of unknowns.
Like the cube test problem, all four curves behave at least as good asas $1 / N$, and the linear Galerkin method is the best, decaying nearly as $1 / N^2$ for larger $N$.

\section{Conclusion}

We have presented a method for computing the double surface integrals encountered in the Galerkin BEM.
When the boundary is discretized using triangular elements, these integrals are performed over pairs of these triangles.
They can be extremely difficult to compute, especially when the two triangles share a vertex, an edge, or are the same.
This is because the kernels being integrated are often singular along the corners and edges of these triangles.
We have solved this problem by using several scaling properties of the integrals and the kernels being integrated.
The integral is broken up into several smaller ones, some of which are written in terms of the original.
This is done in such a way that only completely regular integrals have to be computed explicitly.

We have also presented an analytical method for computing the integrals when the two triangles do not touch.
The method uses spherical harmonics and multipole and local expansions and translations.
The only source of error in this method is how soon to truncate these expansions.
However, the truncation number is adaptively selected to achieve a desired error bound.

Finally, we have implemented the two methods as part of a Galerkin BEM library using MATLAB, and we have made this library freely available for download from our webpage.

\section{Acknowledgements}

Ross Adelman was supported under cooperative agreement W911NF1420118 between the Army Research Laboratory and the University of Maryland, with David Hull and Stephen Vinci as Technical Monitors.
Nail Gumerov and Ramani Duraiswami were partially supported by the same cooperative agreement, and partially by NSF award CMMI1250187.
This report is also available as University of Maryland Department of Computer Science Technical report CS-TR-5043, and Institute of Advanced Computer Studies Technical Report UMIACS-TR-2015-02.

\bibliographystyle{jasanum}
\bibliography{references}

\end{document}